\documentclass[journal]{IEEEtran}
\usepackage{color}
\usepackage[cmex10]{amsmath}
\usepackage{xcolor}
\usepackage[multiple]{footmisc}

\usepackage{mdwtab}
\usepackage{eqparbox}
\usepackage{tabularx}
\usepackage{graphicx,amssymb,rangecite,upgreek,dsfont,mathrsfs}

\renewenvironment{IEEEbiography}[1]
  {\IEEEbiographynophoto{#1}}
  {\endIEEEbiographynophoto}

\newcounter{MYtempeqncnt}

\usepackage{algorithm}
\usepackage{algorithmic} 
\usepackage[usenames,dvipsnames]{pstricks}
 \usepackage{epsfig}

 \usepackage{pst-grad} 
 \usepackage{pst-plot} 

\global\long\def\s[#1]{\textnormal{\scriptsize #1}}

\newcommand {\bxt} {\mbox{\footnotesize\boldmath $x$}}

\newcommand {\byt} {\mbox{\footnotesize\boldmath $y$}}

\newcommand {\bx} {\mbox{\boldmath $x$}}
\newcommand {\by} {\mbox{\boldmath $y$}}

\newcommand {\bE} {\mathbb{E}}

\newcommand {\bX} {\mbox{\boldmath $X$}}
\newcommand {\bY} {\mbox{\boldmath $Y$}}

\newcommand {\bXt} {\mbox{\boldmath \footnotesize $X$}}

\newcommand{\calA}{{\cal A}}

\newcommand{\calC}{{\cal C}}
\newcommand{\calD}{{\cal D}}
\newcommand{\calE}{{\cal E}}

\newcommand{\calG}{{\cal G}}

\newcommand{\calI}{{\cal I}}

\newcommand{\calL}{{\cal L}}

\newcommand{\calQ}{{\cal Q}}
\newcommand{\calR}{{\cal R}}
\newcommand{\calS}{{\cal S}}
\newcommand{\calT}{{\cal T}}
\newcommand{\calU}{{\cal U}}
\newcommand{\calV}{{\cal V}}
\newcommand{\calW}{{\cal W}}
\newcommand{\calX}{{\cal X}}
\newcommand{\calY}{{\cal Y}}

\newcommand{\be}{\begin{equation}}
\newcommand{\ee}{\end{equation}}
\newcommand{\beqna}{\begin{eqnarray}}
\newcommand{\eeqna}{\end{eqnarray}}

\usepackage{theorem}
\DeclareFontFamily{U}{mathx}{\hyphenchar\font45}
\DeclareFontShape{U}{mathx}{m}{n}{
      <5> <6> <7> <8> <9> <10>
      <10.95> <12> <14.4> <17.28> <20.74> <24.88>
      mathx10
      }{}
\DeclareSymbolFont{mathx}{U}{mathx}{m}{n}
\DeclareMathSymbol{\bigtimes}{1}{mathx}{"91}
\theorembodyfont{\rmfamily}

\newcommand{\abs}[1]{\left|#1\right|}

\theoremheaderfont{\itshape}

\newtheorem{theorem}{Theorem}
\newtheorem{proof}{Proof}
\newtheorem{example}{Example}

\newtheorem{lemma}{Lemma} 
\newtheorem{corollary}{Corollary}

\newtheorem{rem}{Remark}
\newcommand{\p}[1]{\left(#1\right)}
\newcommand{\pp}[1]{\left[#1\right]}
\newcommand{\ppp}[1]{\left\{#1\right\}}

\renewcommand\[{\begin{equation}}
\renewcommand\]{\end{equation}}

\begin{document}
\title{Erasure/List Random Coding Error Exponents Are Not Universally Achievable}
\author{Wasim~Huleihel
				~~~~~~~Nir~Weinberger
        ~~~~~~~Neri~Merhav
				\\
        Department of Electrical Engineering \\
Technion - Israel Institute of Technology \\
Haifa 3200003, ISRAEL\\
E-mail: \{wh@campus, nirwein@campus, merhav@ee\}.technion.ac.il
\thanks{This research was partially supported by The Israeli Science Foundation (ISF), grant no. 412/12. This paper was presented in part at the 2015 IEEE Information Theory Workshop (ITW), and the 2015 Information Theory and Applications (ITA) Workshop.}
}
\date{}
\maketitle
\thispagestyle{empty}

\IEEEpeerreviewmaketitle

\begin{abstract}
We study the problem of universal decoding for unknown discrete memoryless channels in the presence of erasure/list option at the decoder, in the random coding regime. Specifically, we harness a universal version of Forney's classical erasure/list decoder developed in earlier studies, which is based on the competitive minimax methodology, and guarantees universal achievability of a certain fraction of the optimum random coding error exponents. In this paper, we derive an exact single-letter expression for the maximum achievable fraction. Examples are given in which the maximal achievable fraction is strictly less than unity, which imply that, in general, there is no universal erasure/list decoder which achieves the same random coding error exponents as the optimal decoder for a known channel. This is in contrast to the situation in ordinary decoding (without the erasure/list option), where optimum exponents are universally achievable, as is well known. It is also demonstrated that previous lower bounds derived for the maximal achievable fraction are not tight in general. We then analyze a generalized random coding ensemble which incorporate a training sequence, in conjunction with a suboptimal practical decoder (``plug-in" decoder), which first estimates the channel using the known training sequence, and then decodes the remaining symbols of the codeword using the estimated channel. One of the implications of our results, is setting the stage for a reasonable criterion of optimal training. Finally, we compare the performance of the ``plug-in" decoder and the universal decoder, in terms of the achievable error exponents, and show that the latter is noticeably better than the former.
\end{abstract}
\section{Introduction}
\IEEEPARstart{I}{n} many practical situations encountered in coded communication systems, the channel over which transmission takes place is unknown to the receiver. Typically, the optimal maximum likelihood (ML) decoder depends on the channel statistics, and therefore its usage is precluded. In such cases, universal decoders are sought which do not require knowledge of the actual channel, but still preform well just as if the channel was known to the decoder. The design of such universal decoders was extensively addressed for ordinary decoding (without the erasure/list option), see, e.g., \cite{Goppa,ZivUni,Csis2,csiszar2011information,NeriUni,UniNeri2,FerderLapidoth}, and references therein. For example, for unknown discrete memoryless channels (DMCs), the maximum mutual information (MMI) decoder \cite{Goppa} is asymptotically optimal for ordinary decoding, in the sense that it achieves the same random coding error exponents as the ML decoder. However, for decoders with an erasure/list option, only partial results exist. 

In this paper, we focus on universal erasure/list decoders proposed and analyzed by Forney for known channels \cite{Forney68}. Erasure/list decoding is especially attractive for unknown channels, since communicating at any fixed rate, however small, is inherently problematic, since this fixed rate might be larger than the unknown capacity of the underlying channel. It makes sense to try to adapt the coding rate to the channel conditions, which can be learned on-line at the transmitter whenever a feedback link from the receiver to the transmitter is available. A possible approach to handle the problem described above is the \emph{rateless coding} methodology, see, for example \cite{Burnashev,ShluRate,rate1,rate2,rate3,rate4}, in which at every time instant the decoder either makes a decision on one of the transmitted messages or decides to request an additional symbol via the feedback line. The latter case can be considered as an ``erasure" event for the decoder, and so universal erasure decoders are required (see discussion in \cite{universal_minimax_erasure}). 

In \cite[Chapter 10, Theorem 10.11]{csiszar2011information}, Csisz{\'a}r and K{\"o}rner proposed a family of universal erasure decoders, parametrized by some real parameter, for DMCs, and analyzed the resulting error exponents. While this family is in the spirit of the MMI decoder, it does not achieve the same exponents as Forney's optimal erasure/list decoder. More recently, in \cite{Moulin_universal_erasure}, Moulin has generalized this family of decoders and proposed a family of decoders parametrized by a weighting function. An optimal weighting function was sought which maximizes the total error exponent of the worst channel in the family, under a constraint on the worst channel undetected-error exponent (the worst channel associated with the two exponents might be different). The decoder was considered universal if the above mentioned trade-off between the worst case exponents does not change even if the choice of specific decoder in the family of allowed decoders can depend on the channel (see \cite[Eq. (3.11)]{Moulin_universal_erasure}, and the discussion that follows). However, this is a rather weak criterion, in the sense that the optimal decoder only depends on the worst case exponents. So, if the family of channels is rich enough (e.g. includes channels whose capacity is lower than the required rate), then the worst case exponents are simply zero, and any decoder is universal. To this end, a stronger criterion for universality was proposed, which states that a decoder is universal if it achieves Forney's exponents (for a known channel) for all channels in the family. In \cite[Proposition 5.5]{Moulin_universal_erasure}, Moulin provided sufficient conditions under which the decoder of Csisz{\'a}r and K{\"o}rner is universal in the strong sense. Loosely speaking, it is required that the total error exponent is small enough for all channels in the family. These conditions, however, strongly limit the families of channels for which this decoder is universal.

In \cite{universal_minimax_erasure}, Merhav and Feder studied the problem using a different approach. Specifically, they considered the problem of universal decoding with an erasure/list option for the class of DMCs indexed by an unknown parameter $\theta$. They invoked the competitive minimax methodology proposed in \cite{merFeder}, in order to derive a universal version of Forney's classical erasure/list decoder. Recall that for a given DMC with parameter $\theta$, a given coding rate $R$, and a given threshold parameter $T$ (all to be formally defined later), Forney's erasure/list decoder optimally trades off between the exponent, $E_1(R, T, \theta)$, of the probability of total error event, $\calE_1$, and the exponent, $E_2(R, T, \theta) = E_1(R, T, \theta) + T$, of the probability of undetected error event, $\calE_2$, for an erasure decoder (or, average list size for list decoder), in the random coding regime. The universal erasure/list decoder of \cite{universal_minimax_erasure} guarantees achievability of an exponent, $\hat{E}_1(R, T, \theta)$, which is at least as large as $\xi\cdot E_1(R, T, \theta)$ for all $\theta$, for some constant $\xi\in(0,1]$ that is independent of $\theta$ (but does depend on $R$ and $T$), and at the same time, an undetected error exponent for erasure decoder (or, average list size for list decoder) $\hat{E}_2(R, T, \theta)\geq\xi\cdot \hat{E}_1(R, T, \theta)+T$ for all $\theta$. At the very least this guarantees that whenever the probabilities of $\calE_1$ and $\calE_2$ decay exponentially for a known channel, so they do even when the channel is unknown, using the proposed universal decoder. It should be remarked, that the benchmark exponents in \cite{universal_minimax_erasure} were the classical lower bounds on $E_1(R, T, \theta)$ and $E_2(R, T, \theta)$ derived by Forney \cite{Forney68}. 

Clearly, to maximize the guaranteed exponents obtained by the universal decoder of \cite{universal_minimax_erasure}, the maximal $\xi\in\pp{0,1}$ such that the above holds is of interest. This maximal fraction is the central quantity of this paper and will be denoted henceforth by $\xi^*(R,T)$. If, for example, $\xi^*(R,T)$ is strictly less than unity, then it means that there is a major difference between universal ordinary decoding and universal erasure/list decoding: while for the former, it is well known that optimum random coding error exponents are universally achievable (at least for some classes of channels and certain random coding distributions), in the latter, when the erasure/list options are available, this may no longer be the case\footnote{We could have similarly required that the universal decoder would achieve an undetected error exponent of $\hat{E}_2(R,T,\theta)\geq \tilde{\xi}\cdot E_2(R,T,\theta)$ for all $\theta\in\Theta$, and some $\tilde{\xi}\in(0,1]$. While the numerical value of the maximal achievable $\tilde{\xi}$, say $\tilde{\xi}^*(R,T,\theta)$, will be different from $\xi^*(R,T,\theta)$, the main conclusions of the paper will not change. Specifically, $\xi^*(R,T,\theta)<1$ if and only if $\tilde{\xi}^*(R,T,\theta)<1$.}. In \cite{universal_minimax_erasure}, Merhav and Feder invoked Gallager's bounding techniques to analyze the exponential behavior of upper bounds on the probabilities $\calE_1$ and $\calE_2$. Accordingly, a single-letter expression for a lower bound to $\xi^*(R,T)$ was obtained, which we denote henceforth by $\xi_L(R,T)$. Since $\xi_L(R,T)$ was merely a lower bound, the question of achievability of Forney's erasure/list exponents was not fully settled in \cite{universal_minimax_erasure}\footnote{Note that universality in the weak sense in \cite{Moulin_universal_erasure} does not guarantee that $\xi^*(R,T)$ is larger than zero because this weak criterion only considers the worst case channels. A universal decoder in the stronger sense in \cite{Moulin_universal_erasure} does imply that $\xi^*(R,T)=1$, but, as previously mentioned, such universality was proved only for a restricted families of channels.}. 

As was previously mentioned, even for a known channel, only lower bounds for the exponents were obtained by Forney \cite{Forney68}. More recently, inspired by a statistical-mechanical point of view on random code ensembles, Somekh-Baruch and Merhav \cite{exact_erasure} have found \emph{exact} expressions for the exponents of the optimal erasure/list decoder, by assessing the moments of certain type class enumerators. In this paper, we tackle again the problem of erasure/list channel decoding using similar methods, and derive an \emph{exact} expression for $\xi^*(R,T)$ with respect to the \emph{exact} erasure/list exponents of a known channels found in \cite{exact_erasure}. Unlike the lower bound of \cite{universal_minimax_erasure}, the exact expression leads to the following conclusions:

\begin{enumerate}
\item In general, $\xi^*(R,T)$ is strictly less than $1$. Therefore, the known channel exponents in erasure/list decoding cannot be achieved universally. In this sense, channel knowledge is crucial for asymptotically optimum erasure/list decoding. This is in sharp contrast to the situation in ordinary decoding (without the erasure/list option), where, as said, optimum exponents are universally achievable, e.g., by the MMI decoder.
\item In general, $\xi_L(R,T)$ is strictly less than $\xi^*(R,T)$. Therefore, the Gallager-style analysis technique in \cite{universal_minimax_erasure} is not always powerful enough to obtain $\xi^*(R,T)$.
\end{enumerate}

Although the above universal decoder achieves $\xi^*(R,T)$, it may have a rather high implementation complexity. Usually, in practical communication systems with channel uncertainty, a portion of the blocklength is devoted to training which is a common part of all codewords. A possible practical decoder is the ``plug-in" decoder, which first estimates the channel using the known training sequence, and then decodes the remaining symbols of the codeword using the estimated channel from the first stage. This suboptimal decoder, on the one hand, has a smaller complexity, and thus can be more easily incorporated into practical systems, but on the other hand, achieves only some $\xi^e(R,T)\leq\xi^*(R,T)$. For this sub-optimal decoder, we derive its error exponents and a closed-form formula for $\xi^e(R,T)$, which now depend also on the relative training time and the type of the sequence. One implication of our results, is setting the stage for a reasonable criterion of optimal training. Finally, we show numerically that there is a noticeable loss in the error exponents incurred by the plug-in decoder compared to the universal decoder.

The outline of the rest of  the paper is as follows. In Section \ref{sec:not}, we establish notation conventions, and in Section \ref{sec:back} we detail necessary background on erasure/list decoding, both for known and unknown channels. Then, in Section \ref{sec:Main}, we present our main result of an exact expression for $\xi^*(R,T)$, and discuss the special case of binary symmetric channel (BSC). We then shed light on the differences between $\xi^*(R,T)$ and $\xi_L(R,T)$, along with some numerical results, which illustrate the main result of this paper. In Section \ref{sec:training}, we analyze generalized random coding ensembles which incorporates a training sequence, in conjunction with the suboptimal plug-in decoder and the universal decoder, and compare its performance with the universal decoder and the optimal decoder (for known channel). Finally, in Section \ref{sec:proofs}, we provide proofs for all our results. 
\allowdisplaybreaks

\section{Notation Conventions}\label{sec:not}

Throughout this paper, scalar random variables (RVs) will be denoted by capital letters, their sample values will be denoted by the respective lower case letters, and their alphabets will be denoted by the respective calligraphic letters, e.g. $X$, $x$, and $\calX$, respectively. A similar convention will apply to random vectors of dimension $n$ and their sample values, which will be denoted with the same symbols in the boldface font. The set of all $n$-vectors with components taking values in a certain finite alphabet, will be denoted as the same alphabet superscripted by $n$, e.g., $\calX^n$. Generic channels will be usually denoted by the letters $P$, $Q$, or $W$. We shall mainly consider joint distributions of two RVs $(X,Y)$ over the Cartesian product of two finite alphabets $\calX$ and $\calY$. For brevity, we will denote any joint distribution, e.g. $Q_{XY}$, simply by $Q$, the marginals will be denoted by $Q_X$ and $Q_Y$, and the conditional distributions will be denoted by $Q_{X\vert Y}$ and $Q_{Y\vert X}$. The joint distribution induced by $Q_{X}$ and $Q_{Y|X}$ will be denoted by $Q_{X}\times Q_{Y|X}$, and a similar notation will be used when the roles of $X$ and $Y$ are switched. 

The expectation operator will be denoted by $\bE\ppp{\cdot}$, and when we wish to make the dependence on the underlying distribution $Q$ clear, we denote it by $\bE_Q\ppp{\cdot}$. The entropy of $X$ and the conditional entropy of $X$ given $Y$, will be denoted $H_X(Q)$, $H_{X\vert Y}(Q)$, respectively, where $Q$ is the underlying probability distribution. The mutual information of the joint distribution $Q$ will be denoted by $I(Q)$. The divergence (or, Kullback-Liebler distance) between two probability measures $Q$ and $P$ will be denoted by $D(Q||P)$. For two numbers $0\leq q,p\leq 1$, $D(q||p)$ will stand for the divergence between the binary measures $\ppp{q,1-q}$ and $\ppp{p,1-p}$. 

For a given vector $\bx$, let $\hat{Q}_{\bxt}$ denote the empirical distribution, that is, the vector $\{\hat{Q}_{\bxt}(x),~x\in{\cal X}\}$, where $\hat{Q}_{\bxt}(x)$ is the relative frequency of the letter $x$ in the vector $\bx$. Let ${\cal T}_{P}$ denote the type class associated with $P$, that is, the set of all sequences $\bx$ for which $\hat{Q}_{\bxt}=P$. Similarly, for a pair of vectors $(\bx,\by)$, the empirical joint distribution will be denoted by $\hat{Q}_{\bxt\byt}$, or simply by $\hat{Q}$, for short. All the previously defined notations for regular distributions will also be used for empirical distributions.

The cardinality of a finite set $\calA$ will be denoted by $\abs{\calA}$, its complement will be denoted by $\calA^c$. The probability of an event $\calE$ will be denoted by $\Pr\ppp{\calE}$. The indicator function of an event $\calE$ will be denoted by $\calI\ppp{\calE}$. For two sequences of positive numbers, $\ppp{a_n}$ and $\ppp{b_n}$, the notation $a_n\doteq b_n$ means that $\ppp{a_n}$ and $\ppp{b_n}$ are of the same exponential order, i.e., $n^{-1}\log a_n/b_n\to0$ as $n\to\infty$, where in this paper, logarithms are defined with respect to (w.r.t.) the natural basis, that is, $\log(\cdot)\equiv\ln(\cdot)$. Finally, for a real number $x$, we let $\abs{x}^+ \triangleq \max\ppp{0,x}$.

\section{Model Formulation and Short Background}\label{sec:back}
\subsection{Known Channel}\label{subsec:back1}
Consider a DMC with a finite input alphabet $\calX$, finite output alphabet $\calY$, and a matrix of single-letter transition probabilities $\ppp{W\p{y\vert x},\ x\in\calX,\;y\in\calY}$. A rate-$R$ codebook consists of $M=\left\lceil e^{ nR}\right\rceil$ length-$n$ codewords $\bx_m\in\calX^n$, $m=1,2,\ldots,M$, representing the $M$ messages. It will be assumed that all messages are a-priori equiprobable. We assume the ensemble of fixed composition random codes of blocklength $n$, where each codeword is selected at random, uniformly within a type class $\calT(P_X)$ for some given random coding distribution $P_X$ over the alphabet $\calX$.

In the following, we give a short description on the operation of the erasure decoder and then the list decoder. A decoder with an erasure option is a partition of the observation space $\calY^n$ into $\p{M+1}$ regions, denoted by $\ppp{\calR_m}_{m=0}^M$. An erasure decoder works as follows: If $\by\in\calY^n$ falls into the $m$th region, $\calR_m$, for $m=1,2,\ldots,M$, then a decision is made in favor of message number $m$. If $\by\in\calR_0$, then no decision is made and an erasure is declared. Accordingly, we shall refer to $\by\in\calR_0$ as an \emph{erasure event}. Given a code $\calC\triangleq\ppp{\bx_1,\ldots,\bx_M}$ and a decoder $\calR \triangleq \p{\calR_0,\ldots,\calR_M}$, we define two error events. The event $\calE_1$ is the event of deciding on erroneous codeword or making an erasure, and the event $\calE_2$ which is the undetected error event, namely, the event of deciding on erroneous codeword. It is evident that $\calE_1$ is the disjoint union of the erasure event and $\calE_2$. The probabilities of all the aforementioned events are given by:
\begin{align}
&\Pr\ppp{\calE_1} = \frac{1}{M}\sum_{m=1}^M\sum_{\byt\in\calR_m^c}W\p{\by\vert\bx_m},\label{E1Def}\\
&\Pr\ppp{\calE_2} = \frac{1}{M}\sum_{m=1}^M\sum_{\byt\in\calR_m}\sum_{m'\neq m}W\p{\by\vert\bx_{m'}},\label{E2Def}
\end{align}
and
\begin{align}
\Pr\ppp{\calR_0} = \Pr\ppp{\calE_1}-\Pr\ppp{\calE_2}.
\end{align}

A list decoder is a mapping from the space of received vectors ${\cal Y}^{n}$ into a collection of the subsets of $\{1,\ldots,M\}$. Alternatively, a list decoder is uniquely defined by a set of $M+1$ (not necessarily disjoint) decoding regions $\{{\cal R}_{m}\}_{m=0}^{M}$ such that ${\cal R}_{m}\subseteq{\cal Y}^{n}$ and ${\cal R}_{0}={\cal Y}^{n}\backslash\bigcup_{m=1}^{M}{\cal R}_{m}$. Given a received vector $\by$, the $m$th codeword belongs to the output list if $\by\in{\cal R}_{m}$, and if $\by$ does not belong to any of the regions ${\cal R}_{m}$ then $\by\in{\cal R}_{0}$, and an erasure is declared. The average error probability of a list decoder and a codebook ${\cal C}$ is the probability that the actual transmitted codeword does not belong to the output list, and it is defined similarly to \eqref{E1Def}. The average list size is the expected (w.r.t. the output of the channel) number of erroneous codewords in the output list, and it is easily verified that it is defined exactly as in \eqref{E2Def} (see  \cite[Eq. (13)]{Forney68}).

Since the error events for the erasure and list decoders are defined in the same way, they can be treated on the same footing. Nonetheless, for descriptive purposes, we will refer to the erasure decoder, but we emphasize that all the following analysis and results are true also for the list decoder. When knowledge on the specific DMC is available at the decoder, Forney has shown in \cite{Forney68}, using the Neyman-Pearson methodology, that the optimal trade-off between $\Pr\ppp{\calE_1}$ and $\Pr\ppp{\calE_2}$ is attained by the decision regions $\calR^*\triangleq\p{\calR_0^*,\ldots,\calR_M^*}$ given by:
\begin{align}
{\cal R}^*_{m}\triangleq\ppp{\by:\;\frac{W\p{\by\vert\bx_m}}{\sum_{m'\neq m}W\p{\by\vert\bx_{m'}}}\geq e^{nT}},\ m=1,2,\ldots M,\label{optDD1}
\end{align}
and
\begin{align}
\calR_{0}^*\triangleq\bigcap_{m=1}^{M}\p{\calR_{m}^*}^{c},\label{optDD2}
\end{align}
where $T$ is a parameter, henceforth referred as the \emph{threshold}, which controls the balance between the probabilities of $\calE_1$ and $\calE_2$. When $T\geq0$ the decoder operates in the erasure mode, and when it is in the list mode then $T<0$. No other decision rule gives both a lower $\Pr\ppp{\calE_1}$ and a lower $\Pr\ppp{\calE_2}$ than the above choice. Finally, we define the error exponents $E_i\p{R,T},\;i=1,2$, as the exponents of the average probabilities of errors $\overline{\Pr}\ppp{\calE_i}$ (associated with the optimal decoder $\calR^*$), where the average is taken w.r.t. a given ensemble of the randomly selected codes, that is,
\begin{align}
E_i\p{R,T}\triangleq-\liminf_{n\to\infty}\frac{1}{n}\log\overline{\Pr}\ppp{\calE_i},\ i=1,2.\label{errorExpDef}
\end{align}  
An important observation is that Forney's decision rule for known DMCs can also be obtained by formulating the following optimization problem: Find a decoder $\calR$ that minimizes $\Gamma\p{\calC,\calR}$ where 
\begin{align}
\Gamma\p{\calC,\calR}&\triangleq\Pr\ppp{\calE_2}+e^{-nT}\Pr\ppp{\calE_1}\label{gammaTh0}\\
&=\frac{1}{M}\sum_{m=1}^{M}\left[\sum_{\byt\in{\cal R}_{m}}\sum_{m'\neq m}W(\by|\bx_{m'})\right.\nonumber\\
&\left.\ \ \ \ \ \ \ \ \ \ \ \ \ \ \ \ \ +\sum_{\byt\in{\cal R}_{m}^{c}}e^{-nT}W(\by|\bx_{m})\right]\label{gammaTh}
\end{align}
for a given codebook $\calC$ and a given threshold $T$. Indeed, noting that \eqref{gammaTh} can be rewritten as
\begin{align}
\Gamma\p{\calC,\calR} &=\sum_{\byt\in{\calY}^n}\frac{1}{M}\sum_{m=1}^{M}\left[\sum_{m'\neq m}W(\by|\bx_{m'})\calI\ppp{\by\in{\cal R}_{m}}\right.\nonumber\\
&\left.\ \ \ \ \ \ \ \ \ \ \ \ \ \ \ \ +\vphantom{\sum_{m'\neq m}}e^{-nT}W(\by|\bx_{m})\calI\ppp{\byt\in{\calR}^c_{m}}\right],
\end{align}
it is evident that for each $m$, the bracketed expression is minimized by $\calR_m^*$ as defined above. By taking the ensemble average, we have
\begin{align}
\bE\ppp{\Gamma\p{\calC,\calR^*}}\triangleq\overline{\Pr}\ppp{\calE_2}+e^{-nT}\overline{\Pr}\ppp{\calE_1}.\label{TwoTerms}
\end{align}
In \cite{exact_erasure}, it was stated (without a proof) that, in the exponential scale, there is a balance between the two terms at the right hand side of \eqref{TwoTerms}, namely, the exponent of $\overline{\Pr}\ppp{\calE_2}$ equals to the exponent of $e^{-nT}\overline{\Pr}\ppp{\calE_1}$, for the optimal decoder $\calR^*$. We rigorously assert this property in the following lemma, the proof of which appears in Appendix \ref{app:1}.
\begin{lemma}\label{lem:1}
For all $R$ and $T$, the optimal decoder $\calR^*$ satisfies:
\begin{align}
E_2\p{R,T} = T+E_1\p{R,T}.
\end{align}
\end{lemma}

The significance of Lemma \ref{lem:1} is attributed to the fact that now we only need to assess the exponential behavior of either $\overline{\Pr}\ppp{\calE_1}$, or, $\overline{\Pr}\ppp{\calE_2}$, but not both. As was mentioned in the Introduction, in \cite{exact_erasure}, Somekh-Baruch and Merhav have obtained exact single-letter formulas for the error exponents $E_1(R,T)$ and $E_2(R,T)$ associated with $\overline{\Pr}\ppp{\calE_1}$ and $\overline{\Pr}\ppp{\calE_2}$, respectively. Specifically, they show, that for the ensemble of fixed composition codes \cite[Theorem 1]{exact_erasure}\footnote{In \cite{exact_erasure}, each codeword in the codebook was drawn independently of all other codewords, and its symbols were drawn from an independent and identically (i.i.d.) distribution (identical for all the codewords). Nonetheless, the modification to the ensemble of fixed composition codes is straightforward.}\footnote{We note that there is an error at the end of the proof of Theorem 1 in \cite{exact_erasure}, where it was claimed that $\min\ppp{E_a(R,T),E_b(R,T)} = E_a(R,T)$, which may not be true in general. The correct expression is as in \eqref{AneliaExp}.}:
\begin{align}
E_1(R,T) = \min\ppp{E_a(R,T),E_b(R,T)},\label{AneliaExp}
\end{align}
where
\begin{align}
E_a(R,T) \triangleq \min_{(Q,\tilde{Q})\in\hat{\calQ}}\pp{D(\tilde{Q}||P_X\times W)+\;I(Q)-R}\label{AneliaExp1}
\end{align}
and
\begin{align}
E_b(R,T) \triangleq \min_{\tilde{Q}\in\hat{\calL}}\;D(\tilde{Q}||P_X\times W)\label{AneliaExp2}
\end{align}
where $\tilde{Q}$ is a probability distribution on $\calX\times\calY$, and
\begin{align}
&\hat{\calQ}\triangleq\ppp{(Q,\tilde{Q})\in\calD:\;I(Q)\geq R,\;\hat{\Omega}(Q,\tilde{Q})\leq0},\\
&\calD\triangleq\ppp{(Q,\tilde{Q}):\;Q_X = \tilde{Q}_X = P_X,\ Q_Y = \tilde{Q}_Y},\label{calD}\\
&\hat{\Omega}(Q,\tilde{Q})\triangleq \bE_{\tilde{Q}}\log W(Y\vert X) - \bE_{Q}\log W(Y\vert X) -T,
\end{align}
and
\begin{align}
\hat{\calL}&\triangleq \left\{\vphantom{\max_{Q:(Q,\tilde{Q})\in\calD:\;I(Q)\leq R}}\tilde{Q}:\;\bE_{\tilde{Q}}\log W(Y\vert X)\leq R+T\right.\nonumber\\
&\left.+\max_{Q:(Q,\tilde{Q})\in\calD:\;I(Q)\leq R}\;\pp{\bE_Q\log W(Y\vert X)-I(Q)}\right\}.
\end{align}

As a special case, we shall consider in the sequel the problem of universal erasure/list decoding for the BSC, and to this end, we will use the exact expression of $E_1(R,T)$. Accordingly, for the BSC with crossover probability $\theta$, it was shown that \cite[Corollary 2]{exact_erasure}
\begin{align}
E_{1,\text{BSC}}(R,T) = \min\ppp{E_{a,\text{BSC}}(R,T),E_{b,\text{BSC}}(R,T)},\label{BSCane}
\end{align}
where
\begin{align}
&E_{a,\text{BSC}}(R,T)\triangleq \log 2-R\nonumber\\
&\ \ \ \ +\min_{\tilde{q}\in\pp{\theta,\delta_{GV}\p{R}-T/\beta}}\pp{D\p{\tilde{q}||\theta}-h\p{\tilde{q}+\frac{T}{\beta}}},
\end{align}
and 
\begin{align}
E_{b,\text{BSC}}(R,T)\triangleq\min_{\tilde{q}\in\hat{\calL}_{\text{BSC}}}\;D\p{\tilde{q}||\theta},
\end{align}
where $\beta(\theta) \triangleq \log\pp{(1-\theta)/\theta}$, and $\delta_{\text{GV}}(R)$ denote the normalized Gilbert-Varshamov (GV) distance, i.e., the smaller solution, $\delta$, to the equation
\begin{align}
h(\delta) = \log 2 - R,
\end{align}
where $h(\delta)\triangleq -\delta\log\delta - (1-\delta)\log(1-\delta)$ is the binary entropy function, and
\begin{align}
\hat{\calL}_{\text{BSC}}&\triangleq \left\{\vphantom{\max_{q:\;R\geq\log2-h(q)}}\tilde{q}: -\tilde{q}\cdot\beta(\theta)\leq R+T\right.\nonumber\\
&\left.+\max_{q:\;R\geq\log2-h(q)}\pp{-q\cdot\beta(\theta)+h(q)-\log2}\right\}.
\end{align}

\subsection{Unknown Channel}
We now move on to the case of an unknown channel. Consider a family of DMCs 
\begin{align}
\calW_\Theta\triangleq\ppp{W_\theta\p{y\vert x},\ x\in\calX,y\in\calY,\theta\in\Theta},
\end{align}
with a finite input alphabet $\calX$, a finite output alphabet $\calY$, and a matrix of single-letter transition probabilities $\ppp{W_\theta\p{y\vert x}}$, where $\theta$ is a parameter, or the index of the channel in the class, taking values in some set $\Theta$, which may be countable or uncountable. For example, $\theta$ may represent the set of all $\abs{\calX}\cdot\p{\abs{\calY}-1}$ single-letter transition probabilities that define the DMC with the given input and output alphabets. In our problem, the channel is unknown to the receiver designer, and the designer only knows that the channel belongs to the family of channels $\calW_\Theta$, that is, $\theta$ itself is unknown.

When the channel is unknown, the competitive minimax methodology, proposed and developed in \cite{universal_minimax_erasure}, proves useful. Specifically, let $\Gamma_\theta\p{\calC,\calR}$ in \eqref{gammaTh0} designate the above defined Lagrangian, where we now emphasize the dependence on the index of the channel, $\theta$. Similarly, henceforth we shall denote the error exponents in \eqref{errorExpDef} by $E_1(R,T,\theta)$ and $E_2(R,T,\theta)$. Also, let $\bar{\Gamma}_\theta^*\triangleq\bE\ppp{\min_\calR\Gamma_\theta\p{\calC,\calR}}$, which is the ensemble average of the minimum of the above Lagrangian (achieved by Forney's optimum decision rule) w.r.t. the channel $W_\theta\p{y\vert x}$, for a given $\theta$. Note that by Lemma \ref{lem:1}, the exponential order of $\bar{\Gamma}_\theta^*$ is $e^{-n(E_1\p{R,T,\theta}+T)}$. A \emph{competitive minimax decision rule} $\calR$ is one that achieves
\begin{align}
\min_\calR\max_{\theta\in\Theta}\frac{\Gamma_\theta\p{\calC,\calR}}{\bar{\Gamma}_\theta^*},
\end{align}
which is asymptotically equivalent to
\begin{align}
\min_\calR\max_{\theta\in\Theta}\frac{\Gamma_\theta\p{\calC,\calR}}{e^{-n\pp{E_1\p{R,T,\theta}+T}}}.\label{equminmax}
\end{align}
However, as discussed in \cite{universal_minimax_erasure}, such a minimax criterion, of competing with the optimum performance, may be too ambitious, and the value of the minimization problem in \eqref{equminmax} may diverge to infinity for every $R$, as $n\to\infty$. A possible remedy for this situation is to compete with only a fraction $\xi\in\pp{0,1}$ of $E_1\p{R,T,\theta}$, which we would like to choose as large as possible. To wit, we are interested in the competitive minimax criterion
\begin{align}
K_{n}({\cal C})=\min_\calR K_{n}({\cal C},\calR),
\end{align}
in which
\begin{align}
K_{n}({\cal C},\calR)=\max_{\theta\in\Theta}\frac{\Gamma_{\theta}({\cal C},{\cal R})}{e^{-n\left(\xi E_{1}(R,T,\theta)+T\right)}}. \label{Kn(R,C) definition}
\end{align}
Accordingly, for a given rate $R$ and threshold $T$, we wish to find $\xi^*(R,T)$, defined as:
\begin{align} \label{xi star definition}
\xi^*(R,T) \triangleq \sup\ppp{\xi\in\pp{0,1}:\;\limsup_{n\to\infty}\frac{1}{n}\log \bar{K}_n\leq0},
\end{align}
that is, the largest value of $\xi$ such that the ensemble average $\bar{K}_n\triangleq\bE\ppp{K_{n}({\cal C})}$ would \emph{not} grow exponentially fast.

In \cite{universal_minimax_erasure}, the following universal decoding metric was defined 
\[ \label{ def of f}
f(\bx_{m},\by)\triangleq\max_{\theta\in\Theta}\left\{ e^{n\left[\xi E_{1}(R,T,\theta)+T\right]}W_{\theta}(\by|\bx_{m})\right\}, 
\]
and a universal erasure/list decoder was proposed which has the following decision regions
\[
\hat{{\cal R}}_{m}\triangleq\left\{ \by:\frac{f(\bx_{m},\by)}{\sum_{m'\neq m}f(\bx_{m'},\by)}\geq e^{nT}\right\},\ m=1,2,\ldots M, \label{universal_decoder}
\]
and
\[
\hat{{\cal R}}_{0}\triangleq\bigcap_{m=1}^{M}\hat{{\cal R}}_{m}^{c}.
\]
The property that makes $\hat{\calR}\triangleq(\hat{\calR}_0,\hat{{\cal R}}_{1},\ldots,\hat{{\cal R}}_{M})$ interesting is that it was shown in \cite{universal_minimax_erasure}, that it is asymptotically optimal, i.e., for any given $\xi$, $K_{n}({\cal C},\hat{{\cal R}})$ may only be sub-exponentially larger than $K_n(\calC)$. Thus, the largest $\xi$ such that $\bar{K}_n$ is sub-exponential is also attained by $\hat{{\cal R}}$. Hence, in order to find the largest achievable $\xi$, we would like to evaluate exactly the exponential order of $\bE[K_{n}({\cal C},\hat{{\cal R}})]$, as a function of $\xi$. 

We conclude this section with a few remarks:
\begin{enumerate}
\item Note that the results in this paper can be generalized to other random coding ensembles which assign equal probabilities within every type class (for more details see \cite[Section V]{universal_minimax_erasure}). For conceptual simplicity, we confine attention to fixed-composition random coding.
\item We have assumed that the input distribution $P_X$ is fixed, and so the dependence of $\xi^*(R,T)$ in $P_X$ was omitted. While, in essence, the input distribution may be optimized to maximize $\xi^*(R,T,P_X)$ over some set of input distributions (where, for the moment, we make the dependence in $P_X$ explicit), the meaning of the resulting maximal value should be examined very carefully. Specifically, if we maximize over the entire simplex, the resulting $\max_{P_X}\xi^*(R,T,P_X)$ is simply $1$, which is achieved, trivially and uninterestingly, by any input assignment that puts all its mass on a single codeword. Of course the resultant communication system is completely useless. The point is that the minimax criterion is relative (competitive minimax), i.e., it looks at the difference between the ML exponent and the best universally achievable exponent, allowing (among other things) both exponents to be poor.  It seems that any other conceivable approach for universality will also suffer from a difficulty to define a reasonable criterion for a good choice of $P_{X}$.
\item For $T=0$ it can be shown that the exponent \eqref{AneliaExp} coincides with the ordinary random coding exponent. Since the MMI is a universal decoder which achieves the random coding exponent, then clearly any optimal erasure/list universal decoder may only have better exponents, and therefore $\xi^*(R,0)=1$. 
\end{enumerate}

\section{Results}\label{sec:Main}
In this section, our results are presented and discussed. Proofs are relegated to Section \ref{sec:proofs}.
\subsection{Exact formula for the largest achievable fraction}\label{subsec:1}
We start with a few definitions. Let
\begin{align}
&G(R,T,\xi,\tilde{Q}) \triangleq \max_{\theta\in\Theta} \left\{\vphantom{\bE_{\tilde{Q}}}\xi  E_1\p{R,T,\theta}+T\right.\nonumber\\
&\left.\ \ \ \ \ \ \  \ \ \ \ \ \ \ \ \ \ \ \ \ \ \ \ \ \ \ \ \ \ \ \ +\bE_{\tilde{Q}}\log W_\theta(Y\vert X)\right\},\label{Gdef}\\
&\Omega(R,T,\xi,Q,\tilde{Q}) \triangleq G(R,T,\xi,\tilde{Q})-G(R,T,\xi,Q)-T,\label{Odef}
\end{align}
where $E_1\p{R,T,\theta}$ is given in \eqref{AneliaExp}. Finally, let
\begin{align}
\calQ\triangleq \ppp{(Q,\tilde{Q})\in\calD:\;I(Q)\geq R,\ \Omega(R,T,\xi,Q,\tilde{Q})\leq0}\label{Qdef}
\end{align}
and
\begin{align}
\calL&\triangleq\left\{\vphantom{\max_{Q:(Q,\tilde{Q})\in\calD,\;I(Q)\leq R}}\tilde{Q}:G(R,T,\xi,\tilde{Q})\leq R+T\right.\nonumber\\
&\ \ \ +\left.\max_{Q:(Q,\tilde{Q})\in\calD,\;I(Q)\leq R}\pp{G(R,T,\xi,Q)-I(Q)}\right\} .\label{Ldef}
\end{align}
where $\calD$ is defined in \eqref{calD}.
\begin{theorem}\label{Th:1}
Consider the ensemble of fixed composition codes of type $\calT(P_X)$. Then, for any given $T\in\mathbb{R}$ and $R\geq0$, $\xi^*(R,T)$, defined in \eqref{xi star definition}, is equal to the largest number $\xi$ that simultaneously satisfies:
\begin{align}
&\max_{\theta\in\Theta}\left\{\vphantom{\min_{(Q,\tilde{Q})\in\calQ}\ppp{D({\tilde{Q}}||P_X\times W_\theta)+I(Q)-R}}\xi E_1\p{R,T,\theta}-\right.\nonumber\\
&\left.\ \ \ \min_{(Q,\tilde{Q})\in\calQ}\ppp{D({\tilde{Q}}||P_X\times W_\theta)+I(Q)-R}\right\}\leq0,\label{mainEq}
\end{align}
and
\begin{align}
\max_{\theta\in\Theta}\ppp{\xi E_1\p{R,T,\theta}-\min_{\tilde{Q}\in\calL}D({\tilde{Q}}||P_X\times W_\theta)}\leq0.\label{mainEq2}
\end{align}
\end{theorem} 

Notice that in order to find $\xi^*(R,T)$ one can perform a simple line search over the interval $\pp{0,1}$ using the condition in Theorem \ref{Th:1}. Alternatively, in the following corollary, we also propose an analytical single-letter expression for $\xi^*(R,T)$.
\begin{corollary}\label{cor:1}
Let $\calG(\tilde Q_Y)\triangleq\{Q:\;I(Q)\leq R,\;Q_Y = \tilde Q_Y\}$, and define \eqref{hatXi}-\eqref{hatXi2RT}, shown at the top of the next page, where in the optimizations $Q_X=\tilde Q_X = P_X$.
\begin{figure*}[!t]
\normalsize
\setcounter{MYtempeqncnt}{\value{equation}}
\setcounter{equation}{38} 
\begin{align}
\hat{\xi}_1(R,T,Q,\tilde{Q},\theta,\theta_1,\theta_2,\lambda)&\triangleq\frac{D(\tilde{Q}||P_X\times W_\theta)+I(Q)-R+\lambda\bE_{\tilde{Q}}\pp{\log W_{\theta_1}\p{Y\vert X}}-\lambda\bE_{Q}\pp{\log W_{\theta_2}\p{Y\vert X}}  -\lambda T}{E_1\p{R,T,\theta}-\lambda E_1\p{R,T,\theta_1} +\lambda E_1\p{R,T,\theta_2}},\label{hatXi}\\
\xi_1^*\p{R,T}&\triangleq \min_{\theta\in\Theta}\min_{(\tilde{Q},Q)\in\calD,\;Q\in\calG^c(\tilde Q_Y)}\max_{\lambda\geq0}\max_{\theta_1\in\Theta}\min_{\theta_2\in\Theta}\;\hat{\xi}_1(R,T,Q,\tilde{Q},\theta,\theta_1,\theta_2,\lambda),\label{hatXiRT}\\
\hat{\xi}_2(R,T,Q,\tilde{Q},\theta,\theta_1,\theta_2,\lambda)&\triangleq\frac{D(\tilde{Q}||P_X\times W_\theta)+\lambda\bE_{\tilde{Q}}\pp{\log W_{\theta_1}\p{Y\vert X}}-\lambda\bE_{Q}\pp{\log W_{\theta_2}\p{Y\vert X}}  -\lambda \pp{R+T-I(Q)}}{E_1\p{R,T,\theta}-\lambda E_1\p{R,T,\theta_1} +\lambda E_1\p{R,T,\theta_2}},\label{hatXi2}\\
\xi_2^*\p{R,T}&\triangleq \min_{\theta\in\Theta}\min_{\tilde Q}\max_{\lambda\geq0}\max_{\theta_1\in\Theta}\min_{Q\in\calG(\tilde Q_Y)}\min_{\theta_2\in\Theta}\;\hat{\xi}_2(R,T,Q,\tilde{Q},\theta,\theta_1,\theta_2,\lambda),\label{hatXi2RT}
\end{align}
\hrulefill
\vspace*{4pt}
\end{figure*}
%
%
%
Then, 
\begin{align}
\xi^*(R,T)=\min\ppp{\xi_1^*\p{R,T},\xi_2^*\p{R,T}}.
\end{align}
\end{corollary}

For the special case of the BSC, one can simplify the above minimization problems over the joint distributions $(Q,\tilde{Q})$, and obtain instead a one-dimensional minimization problem. Indeed, consider the family of BSCs where the unknown crossover probability $\theta$ belongs to $\Theta = \pp{0,1}$. Recall that (c.f. end of Subsection \ref{subsec:back1}) $\beta(\theta) = \log\pp{(1-\theta)/\theta}$. Define
\begin{align}
\phi(\theta) &\triangleq \frac{\xi E_1\p{R,T,\theta}+\log(1-\theta)+T}{\beta(\theta)}\nonumber\\
&\ \ \ \ -\frac{\max_{\theta'}\ppp{\xi E_1\p{R,T,\theta'}+\log(1-\theta')-\beta(\theta')\cdot\tilde{q}}}{\beta(\theta)}
\end{align}
and 
\begin{align}
q_1^*&\triangleq \max_{\theta\leq1/2}\phi(\theta),\label{q1Def}\\
q_2^*&\triangleq \min_{\theta>1/2}\phi(\theta).\label{q2Def}
\end{align}
Finally, let
\begin{align}g\p{q_1^*,q_2^*}\triangleq
\begin{cases}
\log 2,\ \ \ \ \ \ \ \ \ \ \text{if}\;q_1^*>1/2,\;\text{or},\;q_2^*<1/2,\\
\max\ppp{h\p{q_1^*},h\p{q_2^*}},\ \ \ \ \ \text{otherwise}
\end{cases}\label{gDef}
\end{align}
and
\begin{align}
&\calL_{\text{BSC}}\triangleq \left\{\tilde{q}: \max_{0\leq\theta\leq1}\pp{\xi E_1(R,T,\theta)-\tilde{q}\cdot\beta(\theta)+\log\theta}\leq R\right.\nonumber\\
&\ \ \ \ \ \ \ \ +T+\max_{0\leq\theta\leq1}\left[\xi E_1(R,T,\theta)-\max\ppp{\theta,\delta_{\text{GV}}(R)}\cdot\beta(\theta)\right.\nonumber\\
&\left.\left.\ \ \ \ \ \ \ \ +\log\theta+h(\max\ppp{\theta,\delta_{\text{GV}}(R)})-\log2\right]\vphantom{\max_{0\leq\theta\leq1}}\right\}.
\end{align}
We have the following result.
\begin{corollary}\label{cor:2}
Consider a family of BSCs, where the unknown crossover probability $\theta$ belongs to $\Theta = \pp{0,1}$, and with fixed composition codes of type $P_X=(1/2,1/2)$. Then, $\xi^*(R,T)$ is equal to the largest number $\xi$ that simultaneously satisfies:
\begin{align}
&\max_{0\leq\theta\leq1}\left\{\vphantom{\min_{\tilde{q}}\pp{D\p{\tilde{q}||\theta}+\abs{-g\p{q_1^*,q_2^*}+\log 2-R}^+}}\xi\cdot E_{1,\text{BSC}}(R,T,\theta)\right.\nonumber\\
&\left.-\min_{\tilde{q}}\pp{D\p{\tilde{q}||\theta}+\abs{-g\p{q_1^*,q_2^*}+\log 2-R}^+}\right\}\leq0,
\end{align}
and
\begin{align}
&\max_{0\leq\theta\leq1}\left\{\xi\cdot E_{1,\text{BSC}}(R,T,\theta)-\min_{\tilde{q}\in\calL_{\text{BSC}}}\;D\p{\tilde{q}||\theta}\right\}\leq0,
\end{align}
where $E_{1,\text{BSC}}(R,T,\theta)$ is given in \eqref{BSCane}. 
\end{corollary}

\subsection{Discussion and Comparison with Previous Results}
While in this work we have derived the exact maximal achievable $\xi^{*}(R,T)$ for fixed composition coding of type $P_{X}$, in \cite[Theorem 2]{universal_minimax_erasure}, Merhav and Feder have obtained the following lower bound \cite[Theorem 2]{universal_minimax_erasure}:
\begin{align}
&\xi^{*}(R,T)\geq\xi_{L}(R,T)\triangleq\nonumber\\
&\min_{(\theta,\theta'')\in\Theta^{2}}\max_{0\leq s\leq\rho\leq1}\frac{E(\theta,\theta'',s,\rho)-\rho R-sT}{(1-s)E_{1}(R,T,\theta)+sE_{1}(R,T,\theta'')}\label{NeriMe}
\end{align}
where 
\begin{align}
&E(\theta,\tilde{\theta},s,\rho)\triangleq\nonumber\\
&\min_{Q_{Y}}\left[F(Q_{Y},1-s,\theta)+\rho F(Q_{Y},s/\rho,\theta'')-H(Q_{Y})\right]
\end{align}
and 
\begin{align}
&F(Q_{Y},1-s,\theta)\triangleq\nonumber\\
&\min_{Q_{X|Y}:\:\left(Q_{Y}\times Q_{X|Y}\right)_{X}=P_{X}}\left[I(Q)-\lambda\mathbb{E}{}_{Q}\log W_\theta(Y\vert X)\right].
\end{align}
Before we continue, we remark that in \cite{universal_minimax_erasure}, Forney's lower bound on $E_{1}(R,T,\theta)$ was used instead of its exact value as derived in \cite{exact_erasure}, but for the sake
of comparison any exponent can be used, and specifically, the exact exponent. Now, note that an alternative (equivalent) representation of $\xi_{L}(R,T)$ in \eqref{NeriMe} is that it is given by the largest $\xi$ such that for any pair $(\theta,\theta'')\in\Theta^{2}$
\begin{align}
\max_{0\leq s\leq\rho\leq1}E(\theta,\theta'',s,\rho)&-\rho R-sT-\xi\left[(1-s)E_{1}(R,T,\theta)\right.\nonumber\\
&\left.\ \ \ \ \ \ \ +sE_{1}(R,T,\theta'')\right]\geq0.
\end{align}
Straightforward algebraic manipulations show that the last inequality can be rewritten
as 
\begin{equation}
\max_{0\leq s\leq\rho\leq1}\min_{(Q,\tilde{Q})\in{\cal D}}\Psi(R,T,\theta,\theta,\theta'',Q,\tilde{Q},\rho,s)\geq0\label{eq: xi achievable condition Neri}
\end{equation}
where
\begin{align}
&\Psi(R,T,\theta,\theta',\theta'',Q,\tilde{Q},\rho,s,\xi)\triangleq D(\tilde{Q}||P_{X}\times W_{\theta})\nonumber\\
&+\rho\left[I(Q)-R\right]+s\cdot\left[\mathbb{E}{}_{\tilde{Q}}\log W_{\theta'}(Y\vert X)+\xi E_{1}(R,T,\theta')\right.\nonumber\\
&\left.\vphantom{\mathbb{E}{}_{\tilde{Q}}\log W_{\theta'}(Y\vert X)}-\mathbb{E}{}_{Q}\log W_{\theta''}(Y\vert X)-\xi E_{1}(R,T,\theta'')-T\right]-\xi E_{1}(R,T,\theta).
\end{align}
For any given $(\theta,\theta'')\in\Theta^{2}$, and $(s,\rho)$, $\Psi(R,T,\theta,\theta',\theta'',Q,\tilde{Q},\rho,s,\xi)$ is convex in\footnote{The input distributions of both $Q$ and $\tilde{Q}$ are assumed fixed to $P_X$, and we are essentially only optimizing over the conditional distributions $(Q_{Y|X},\tilde{Q}_{Y| X})$.} $(Q,\tilde{Q})$, and for a given $(Q,\tilde{Q})$, it is linear (and hence concave) in $(s,\rho)$ . Thus, the minimax theorem implies that \eqref{eq: xi achievable condition Neri} is equivalent to 
\begin{equation}
\min_{(\theta,\theta'')\in\Theta^{2}}\min_{(Q,\tilde{Q})\in{\cal D}}\max_{0\leq s\leq\rho\leq1}\Psi(R,T,\theta,\theta,\theta'',Q,\tilde{Q},\rho,s,\xi)\geq0.\label{eq: xi achievable condition Neri ver 2}
\end{equation}
On the other hand, the exact value of $\xi^{*}(R,T)$ in Theorem \ref{Th:1} is determined by the two conditions \eqref{mainEq}-\eqref{mainEq2}. In what follows, we shall concentrate on the first condition in \eqref{mainEq}, as this condition can be compared to \eqref{eq: xi achievable condition Neri ver 2}. Thus, let us focus on the case in which the condition in \eqref{mainEq} is more stringent than the condition in \eqref{mainEq2}. Then, according to \eqref{mainEq}, a fraction $\xi$ is achievable if 
\begin{equation}
\min_{\theta\in\Theta}\min_{(Q,\tilde{Q})\in{\cal D}}D(\tilde{Q}||P_{X}\times W_{\theta})+I(Q)-R-\xi E_{1}(R,T,\theta)\geq0\label{eq: xi achievable condition ver 1}
\end{equation}
where the minimum over $(Q,\tilde{Q)}$ is such that $I(Q)\geq R$ and $\Omega(R,T,\xi,Q,\tilde{Q})\leq0$. Now, the optimization problem in \eqref{eq: xi achievable condition ver 1} is equivalent to
\begin{align}
&\min_{\theta\in\Theta}\min_{(Q,\tilde{Q})\in{\cal D}}\max_{\rho'\geq0}\max_{s\geq0}\;\left[D(\tilde{Q}||P_{X}\times W_{\theta})-\xi E_{1}(R,T,\theta)\right.\nonumber\\
&\left.\ \ \ \ +(1-\rho')\left[I(Q)-R\right]+s\Omega(R,T,\xi,Q,\tilde{Q})\right]\geq0,
\end{align}
or by letting $\rho=1-\rho'$ we get
\begin{align}
&\min_{\theta\in\Theta}\min_{(Q,\tilde{Q})\in{\cal D}}\max_{\rho\leq1}\max_{s\geq0}\left[D(\tilde{Q}||P_{X}\times W_{\theta})+\rho\left[I(Q)-R\right]\right.\nonumber\\
&\left.\ \ \ \ \ \ \ \ \ \ \ \ +s\Omega(R,T,\xi,Q,\tilde{Q})-\xi E_{1}(R,T,\theta)\right]\geq0,\label{eq: xi achievable condition ver 2}
\end{align}
which is equivalent to
\begin{align}
&\min_{\theta\in\Theta}\min_{(Q,\tilde{Q})\in{\cal D}}\max_{\rho\leq1}\max_{s\geq0}\max_{\theta'\in\Theta}\nonumber\\
&\ \ \ \ \ \ \ \ \ \ \min_{\theta''\in\Theta}\Psi(R,T,\theta,\theta',\theta'',Q,\tilde{Q},\rho,s,\xi)\geq0.
\end{align}
Moreover, for a given $(\theta,Q,\tilde{Q})$, we may write 
\begin{align}
&\max_{\rho\leq1}\max_{s\geq0}\max_{\theta'\in\Theta}\min_{\theta''\in\Theta}\Psi(R,T,\theta,\theta',\theta'',Q,\tilde{Q},\rho,s,\xi)\nonumber\\
&=\min_{\theta''\in\Theta}\max_{\theta'\in\Theta}\max_{0\leq\rho\leq1}\max_{s\geq0}\Psi(R,T,\theta,\theta',\theta'',Q,\tilde{Q},\rho,s,\xi),
\end{align}
because under the constraint $s\geq0$, the inner minimization over $\theta''\in\Theta$ does not depend on the value of $(\rho,s,\theta')$: it is simply the $\theta''\in\Theta$ which maximizes $\mathbb{E}{}_{Q}\log W_{\theta''}(Y\vert X)+\xi E_{1}(R,T,\theta'')$ \footnote{If for a given real function $f(u,v)$ the minimizer $v^*$ w.r.t. $v$ does not depend on $u$, then $\max_{u\in\calU}\min_{v\in\calV}f(u,v)=\max_{u\in\calU}f(u,v^{*})\geq\min_{v\in\calV}\max_{u\in\calU}f(u,v)$, and the minimax inequality results $\max_{u\in\calU}\min_{v\in\calV}f(u,v)=\min_{v\in\calV}\max_{u\in\calU}f(u,v)$, assuming that $\calU$ and $\calV$ are two independent sets (i.e., rectangular).}. Thus, the resulting condition is
\begin{align}
&\min_{(\theta,\theta'')\in\Theta^{2}}\min_{(Q,\tilde{Q})\in{\cal D}}\max_{0\leq\rho\leq1}\max_{s\geq0}\nonumber\\
&\ \ \ \ \ \ \ \ \ \ \ \max_{\theta'\in\Theta}\Psi(R,T,\theta,\theta',\theta'',Q,\tilde{Q},\rho,s,\xi)\geq0.\label{eq: xi achievable condition ver 4}
\end{align}
By comparing the condition in \eqref{eq: xi achievable condition ver 4} to the condition of the lower bound of \cite{universal_minimax_erasure} in \eqref{eq: xi achievable condition Neri ver 2}, the following differences are observed:
\begin{enumerate}
\item In \eqref{eq: xi achievable condition Neri ver 2} an additional constraint
$s\leq\rho$ is imposed.
\item In \eqref{eq: xi achievable condition Neri ver 2} a sub-optimal choice
of $\theta'=\theta$ is imposed.
\end{enumerate}
Accordingly, these differences may cause the value of the minimax in \eqref{eq: xi achievable condition Neri ver 2} to be lower than the value of the optimization problem in \eqref{eq: xi achievable condition ver 4}, which results in a lower achievable $\xi$ compared to $\xi^{*}(R,T)$, as one should expect. Next, we provide two examples, where in the first one these differences are immaterial, and in the second one they do matter. The former happens when the optimal solution in \eqref{eq: xi achievable condition ver 4}, denoted by $(\theta^{*},\theta''^{*},Q^{*},\tilde{Q}^{*},\rho^{*},s^{*})$, satisfies $s^{*}\leq\rho^{*}$, and the maximizer of $\mathbb{E}_{\tilde{Q}^{*}}\log W_{\theta'}(Y\vert X)+\xi_{L}(R,T)\cdot E_{1}(R,T,\theta')$ is given by $\theta^*$. Accordingly, in this case, the value of \eqref{eq: xi achievable condition ver 4} equals to \eqref{eq: xi achievable condition Neri ver 2}. Since, in addition, in this example, the condition in \eqref{mainEq} is more stringent than the condition in \eqref{mainEq2}, we obtain $\xi^{*}(R,T)=\xi_{L}(R,T)$. The conclusion that stems from this observation is that, in this case, the analysis in \cite{universal_minimax_erasure} is tight. 
\begin{example}\label{exmp:1}
In \cite{universal_minimax_erasure}, a family of BSCs was considered where $\theta\in\Theta$ designates the cross-over probability of the BSC, and $\Theta=\{0,1/100,2/100,\ldots,1\}$. The values of $\xi_{L}(R,T)$ were computed for various values of $R$ and $T$. It was assumed that $T\geq0$, which means that the decoder operates in the erasure mode. Numerical calculations of the bound derived in this work (and the exact formula), result in exactly the same values as given in \cite[Table 1]{universal_minimax_erasure}, and so in all these cases, the analysis of \cite{universal_minimax_erasure} was sufficient to provide tight results. For example, for $(R,T) = (0.05,0.15)$, and codebook type $P_{X}=(1/2,1/2)$, we obtain $\xi_{L}(R,T)=0.495$. Also, the two worst case channels (i.e., the solutions to \eqref{eq: xi achievable condition ver 4}) are $\theta^{*}=0.18$ and $\theta''^{*}=0.22$ while $\theta'^{*}=\theta^{*}$ and $\rho^{*}=0.36>s^{*}=0.185$. So, since $s^{*}<\rho^{*}$ and $\theta'^{*}=\theta^{*}$, the discussion above implies that a tight result is obtained, that is, $\xi^{*}(R,T)=\xi_{L}(R,T)=0.495$. Thus, in the worst case over all $\theta\in\Theta$, the exponent $\hat{E}_1(R,T,\theta)$ is not less than $0.495\cdot E_1(R,T,\theta)$.
\end{example}

Since $\xi^*(R,T)<1$ for some $R$ and $T$, we arrive at the following conclusion: \emph{In general, in the random coding regime of erasure/list decoding, there is no universal decoder which achieves the same error exponent as Forney's decoder for every channel in the class.} This fact is in contrast to ordinary decoding, in which the MMI decoder achieves the exact same error exponent as the ML decoder. In this sense, knowledge of the channel is crucial when erasure/list options are allowed. The possible difficulty of universalizing an erasure decoder is apparent for the BSC: While for ordinary decoding, the optimal detector depends only on whether $\theta\leq1/2$ or $\theta>1/2$ (i.e., minimum distance versus maximum distance decoders, respectively), and thus rather easy to universalize, the optimal erasure decoder depends on the exact value of $\theta$.

Nonetheless, in general, we might have that $\xi_{L}(R,T)$ is strictly less than $\xi^{*}(R,T)$. Again, assume that the condition in \eqref{mainEq} dominates \textbf{$\xi^{*}(R,T)$}. To provide intuition, notice that in \eqref{eq: xi achievable condition ver 4} \emph{triplets }$(\theta,\theta',\theta'')\in\Theta^{3}$ are optimized, in contrast to \eqref{eq: xi achievable condition Neri ver 2}, where only \emph{pairs} of channels $(\theta,\theta'')\in\Theta^{2}$ are optimized. Thus, for a family of only two channels, namely, $|\Theta|=2$, typically (but not necessarily) the second difference above, of imposing the constraint $\theta'=\theta$, is immaterial. Then, the only difference between the conditions in \eqref{eq: xi achievable condition Neri ver 2} and \eqref{eq: xi achievable condition ver 4} is the constraint $s\leq\rho$. Let us assume that this is indeed the case, and let us notice that $s$ can be thought as a Lagrange multiplier for the constraint 
\begin{align}
&\mathbb{E}{}_{\tilde{Q}}\log W_{\theta'}(Y\vert X)+\xi E_{1}(R,T,\theta')-\mathbb{E}{}_{Q}\log W_{\theta''}(Y\vert X)\nonumber\\
&\ \ \ \ \ \ \ \ \ \ \ \ \ \ -\xi E_{1}(R,T,\theta'')-T\leq0.\label{eq: xi achievable condition optimization Omega constraint}
\end{align}
Now, if the constraint, at the optimal solution, is slack, then the optimal Lagrange multiplier is $s^{*}=0$. In this case, the constraint $s\leq\rho$ is immaterial and so \eqref{eq: xi achievable condition Neri ver 2} and \eqref{eq: xi achievable condition ver 4} are exactly the same. However, as we shall see in the sequel, it is possible that $s^{*}>\rho^{*}$ in \eqref{eq: xi achievable condition ver 4}, and then the values of the objective in \eqref{eq: xi achievable condition Neri ver 2} and \eqref{eq: xi achievable condition ver 4} are different. Observing \eqref{eq: xi achievable condition optimization Omega constraint}, it is apparent that as $T$ decreases, and especially in the list mode of $T<0$, the optimal $s^{*}$ of \eqref{eq: xi achievable condition ver 4} increases, perhaps beyond the optimal $\rho^{*}$. Thus, if both $s^{*}>\rho^{*}$ and the condition in \eqref{mainEq} dominates \textbf{$\xi^{*}(R,T)$}, we get that $\xi_{L}(R,T)<\xi^{*}(R,T)$. The following example provides such a simple case. We remark, that such a phenomenon was already observed in a Slepian-Wolf erasure/list decoding scenario, for a known source \cite{NeriSW}. There too, in the list regime of $T<0$, there is a gap between the Forney-style bound and the exact random binning error exponents. 
\begin{example}
Consider a family of two BSCs, where $\Theta=\{0.1,0.15\}$, and a type $P_{X}=(1/2,1/2)$ for the random fixed composition codebook. We take $(R,T) = (0.4,-0.25)$, and since $T<0$, the decoder operates in the list mode. We obtain that $\xi_{L}(R,T)=0.716$ which is strictly less than $\xi^{*}(R,T)=0.727$. In the optimization problem \eqref{eq: xi achievable condition Neri ver 2}, the optimal values are $\rho^{*}=s^{*}=0.231$, while if the constraint $s\leq\rho$ is relaxed, then the optimal values are $s=0.231>\rho=0.217$. The resulting value of the optimization problem is exactly $0.727$, just as $\xi^{*}(R,T)$. Moreover, for this example, the largest achievable $\xi$ which satisfies condition \eqref{mainEq} is the same for condition \eqref{mainEq2}. While the difference between $\xi_{L}(R,T)$ and $\xi^{*}(R,T)$ is not very large, it is nevertheless existent and in more intricate scenarios, the differences might be more significant. 
\end{example}


\section{Decoding With Training}\label{sec:training}

Usually, in practical communication systems with channel uncertainty, a portion of the blocklength is devoted to a training sequence which is common to all codewords. This sequence is aimed for learning the unknown channel. In this section, we will first define random coding ensembles which incorporate a training sequence. Then, we shall propose and compare two decoders for this scenario: the (asymptotically optimal) universal decoder in \eqref{universal_decoder}, and a ``plug-in" decoder, which first estimates the channel using the training sequence, and then decodes the remaining symbols of the codeword using the estimated channel.

\subsection{Definition of training ensembles}
For reasons that will be clear in the sequel, we consider two variants of an ensemble which incorporate a training sequence. In the first ensemble, we fix a portion\footnote{As discussed in \cite[Appendix I]{FerderLapidoth}, achieving the random coding error exponent when using a plug-in decoder with a training sequence of length $\alpha_{n}n$ such that $\alpha_{n}\to0$ is not possible, even for ordinary decoding. In a nutshell, the error exponent of the plug-in decoder is not degraded by the estimation error of the channel only when the length of the training sequence is a linear function of $n$. For this reason, we consider a training sequence of length $\alpha n$, where $\alpha$ is a constant fraction.} $\alpha\in\left[0,1\right)$ of the blocklength $n$. Then, a training sequence\footnote{Henceforth, over-bar will indicate quantities which are related to the training part.} $\bar{\bx}\in\calX^{\alpha n}$ is chosen\footnote{For brevity, integer constraints will be omitted.} within type $\bar P_{X}$, and  $M= e^{ nR}$ codewords $\tilde\bx_m\in\calX^{(1-\alpha) n}$, $m=1,2,\ldots,M$, are selected at random, uniformly within a type class $\calT(P_X)$ for some given random coding distribution $P_X$ over the alphabet $\calX$. The transmitted codewords are then the concatenations $\bx_m=\left(\bar{\bx},\tilde\bx_m\right)$ for $1\leq m \leq M$. In the second ensemble, the blocklength of $M= e^{ nR}$ codewords $\tilde\bx_m\in\calX^{n}$, $m=1,2,\ldots,M$, remains $n$, but the codewords are prefixed with a training sequence of length $\beta n$, where $\beta\geq0$. The later ensemble leads of course to a reduction of the effective rate to $R_{\text{eff}} = R/(1+\beta)$. Since the channel is a DMC, it can be easily verified that only the type of the training sequence $\bar P_{X}$ will affect performance, but not the particular sequence within the type class $\calT(\bar P_{X})$. Evidently, when $\alpha=0$ or $\beta=0$, we revert to the ordinary random coding ensemble. Finally, it is important to emphasize that there is an inherent trade-off in using training (i.e., taking $\alpha>0$ or $\beta>0$): learning time comes at the expense of effective blocklength and vice-versa.

\subsection{Universal decoder}\label{subsec:unidecoder}
Whenever $\xi^*(R,T)<1$, one can hope to improve $\xi^*(R,T)$ by using the training ensemble\footnote{In this subsection, we will describe our results only for the first ensemble (defined by $\alpha$), but similar results can be readily derived for the second ensemble (defined by $\beta$).} defined above with $\alpha>0$, along with the \emph{asymptotically optimal decoder} in \eqref{universal_decoder}. That is, even though the first $\alpha n$ symbols are the same for all codewords, the decoder computes the metric $f(\bx_{m},\by)$ for the entire codeword. With a slight abuse of notation, we may denote the maximal fraction achieved by this decoder as $\xi^*(R,T,\alpha,\bar P_X)$, for $\alpha\in[0,1)$, and then $\alpha=0$ corresponds to the ordinary fixed-composition ensemble, considered in Subsection \ref{subsec:1}. The methods used to prove Theorem \ref{Th:1}, can be generalized to obtain $\xi^*(R,T,\alpha,\bar P_X)$,  and in Appendix \ref{app:2}, a closed-form formula for $\xi^*(R,T,\alpha,\bar P_X)$, with a proof outline, are provided. Nonetheless, we suspect that, in fact, $\xi^*(R,T,\alpha,\bar P_X)$ cannot be improved in this way, namely, choosing $\alpha=0$ is optimal.

To gain intuition for the explanation of this phenomena, we focus on two codewords only, $\tilde\bx_1$ and $\tilde\bx_2$, of length $(1-\alpha)n$. In ordinary decoding for a known channel, the decision on the decoded codeword is made only on the basis of the \emph{order} between the likelihoods of both codewords, i.e., $W_\theta(\tilde\by|\tilde\bx_{1})\lessgtr W_\theta(\tilde\by|\tilde\bx_{2})$. On the other hand, in erasure/list decoding for a known channel, the actual likelihood \emph{values} are of importance due to the multiplication of the competing likelihood by $e^{(1-\alpha)nT}$ (recall that, for example,  the first codeword is selected only if $W_\theta(\tilde\by|\tilde\bx_{1})> e^{(1-\alpha)nT}\cdot W_\theta(\tilde\by|\tilde\bx_{2})$). Now, if a common prefix (training sequence) is added to both codewords (and transmitted over the channel), clearly the likelihood of the first part $W_\theta(\bar\by|\bar\bx)$ is the same for both codewords. Let the combined codewords be $\bx_1=(\bar\bx,\tilde\bx_1)$ and $\bx_2=(\bar\bx,\tilde\bx_1)$, and the combined channel output be $\by=(\bar\by,\tilde\by)$. Then, while the order between the combined likelihoods is preserved $W_\theta(\by|\bx_{1})\lessgtr W_\theta(\by|\bx_{2})$, as the blocklength is now $n$ and not $(1-\alpha)n$, the ratio between the values of the two likelihoods (or its inverse), now has to exceed $e^{nT}$, rather than the smaller value of $e^{(1-\alpha)nT}$, so that erasure will not be decided. 

This occurs also in the case of an unknown channel, namely, for the universal decoder in \eqref{universal_decoder}, and in the extreme cases for which $\alpha$ is close to $1$, it may happen that only erasures are decided, which leads to a zero total error exponent. For small and moderate values of $\alpha$ the total error exponent may not be zero, but is still nonetheless worse than the exponent achieved with $\alpha=0$. For the family of BSCs in Example \ref{exmp:1}, we have numerically verified that $\alpha=0$ for all rates and thresholds. We conjecture that this holds for more general families of channels.

\subsection{Plug-in decoder}
A possible practical decoder (termed \emph{``plug-in" decoder}), for the training ensembles defined, works in two stages: First, the decoder estimates the channel using the known training sequence, and then uses this estimated channel in place of the true (unknown) channel in using Forney's decoder \eqref{optDD1}-\eqref{optDD2}, for the remaining symbols of the codeword. This sub-optimal decoder, and the competitive minimax decoder in \eqref{universal_decoder}, are two extremes. Indeed, the decoder in \eqref{universal_decoder} achieves $\xi^*(R,T)$ but may have rather high implementation complexity. The plug-in decoder, on the one hand, has smaller complexity, and thus can be more easily incorporated into practical systems\footnote{If, e.g., the code has some structure and the decoder for a known channel can be implemented for any $\theta\in\Theta$, then the plug-in decoder for an unknown channel only requires an additional estimation step.}, but on the other hand, achieves only some $\xi^e(R,T)\leq\xi^*(R,T)$ (to be rigorously defined in the sequel). Therefore, if $\xi^e(R,T)\ll\xi^*(R,T)$ then there is substantial motivation to use the more complex decoder \eqref{universal_decoder}. If, however, $\xi^e(R,T)\approx\xi^*(R,T)$ then the plug-in decoder is sufficient to almost achieves the optimal performance, while still keeping a reasonable implementation complexity. In this subsection, we analyze the competitive minimax performance of the plug-in decoder.

As mentioned above, the training part, $\bar\bx$, shall be used by the decoder to estimate the channel (this is the first stage). Let us split the output vector $\by$ into two parts $(\bar\by, \tilde\by)$, where the first part corresponds to the training. The channel estimator is a function $\hat{\theta}(\bar Q)\in\Theta$, where $\bar Q=\hat Q_{\bar\bxt \bar\byt}$. Then, in the second stage, optimal decoding (for a known channel) is employed for the remaining symbols of the codeword, assuming that the channel is $W_{\hat{\theta}(\bar Q)}$. Let us denote this plug-in decoder by $\calR^e$, and its associated exponents by $E^e_i(R,T,\theta,\alpha)$, for $\;i=1,2$ for the first ensemble, and $\breve E^e_i(R,T,\theta,\beta)$ for the second ensemble. To analyze these exponents let $E^m_i(R,T,\theta,\hat{\theta})$ for $i=1,2$, designate the error exponents associated with the optimal decoder for a known channel, when tuned to the channel $W_{\hat\theta(\bar Q)}$, but used over the channel $W_\theta$ (i.e., \emph{mismatched decoder}), for the ordinary fixed-composition ensemble (without training). Then, a routine method of types argument reveals that 
\begin{align}
E^e_i(R,T,\theta,\alpha) &= \min_{\bar Q:\; \bar Q_X = \bar P_X}\left\{\vphantom{E^m_i(R,T,\theta,\hat{\theta}(\bar Q))}\alpha \cdot D\p{\bar Q||\bar P_X\times W_\theta}\right.\nonumber\\
&\left.\ \ \ \ \ \ \ \ + (1-\alpha) \cdot E^m_i(R,T,\theta,\hat{\theta}(\bar Q))\right\},\label{ExponentPlugIn1}\\
\breve E^e_i(R,T,\theta,\beta) &= \min_{\bar Q:\; \bar Q_X = \bar P_X}\left\{\vphantom{E^m_i(R,T,\theta,\hat{\theta}(\bar Q))}\beta \cdot D\p{\bar Q||\bar P_X\times W_\theta}\right.\nonumber\\
&\left.\ \ \ \ \ \ \ \ \ \ \ \ \ \ \ \ \ + E^m_i(R,T,\theta,\hat{\theta}(\bar Q))\right\},\label{ExponentPlugIn2}
\end{align}
for $i=1,2$. Now, $E^m_i(R,T,\theta,\hat{\theta})$ can be obtained by simply replacing every instance of $\bE_{Q}{\log W_{\theta}(Y|X)}$, which represent the log-likelihoods assuming the correct channel, with the mismatched log-likelihoods  $\bE_{Q}{\log W_{\hat\theta(\bar Q)}(Y|X)}$ in the exponent expressions of \cite[Theorem 1  and Theorem 2]{exact_erasure}\footnote{As mentioned before, in \cite{exact_erasure} the i.i.d. ensemble was assumed. The modification to the fixed-composition ensemble is straightforward, and only requires removing the $D(Q_X||P_X)$ terms.}. Note, however, that since a mismatched decoder is, in general, sub-optimal, Lemma \ref{lem:1} \emph{cannot} be used, and the equality $E^m_2(R,T,\theta,\hat{\theta})=E^m_1(R,T,\theta,\hat{\theta})+T$ may not necessarily hold. Thus, in the mismatched case, the expression for $E^m_2(R,T,\theta,\theta)$ (see, \cite[Theorem 2]{exact_erasure}) must be used, along with the above replacement (to obtain $E^m_2(R,T,\theta,\hat{\theta})$). It should be stressed, however, that the expression for $E^m_2(R,T,\theta,\theta)$ in \cite[Theorem 2]{exact_erasure} is valid only for the erasure mode\footnote{In general, the undetected error probability event (pertaining to the error exponent $E_{2}$), is more difficult to analyze than the total error event (pertaining to the error exponent $E_{1}$), and in \cite{exact_erasure}, $E_{2}$ was only analyzed for the erasure regime. The difficulty stems from the fact that the analysis in \cite{exact_erasure} is possible only for disjoint decoding regions, which is not the case in the list regime. Unfortunately, a direct analysis (namely, without relying on the relation $E_{2}=E_{1}+T$, which might be wrong for the plug-in decoder) of the undetected error exponent in the list regime is much more challenging.}, i.e., $T\geq0$, which shall be assumed henceforth. Finally, as can be seen from the above expressions, we need to define/find the estimator $\hat{\theta}(\bar Q)$. If, e.g., $\Theta$ is the family of all DMCs, with input alphabet $\calX$ and output alphabet $\calY$, then the \emph{maximum likelihood estimator} can be used, which in this case, is just the parameter $\theta$ which corresponds to $\bar Q_{Y|X}$ where $\bar Q=\hat Q_{\bar\bxt\bar\byt}$. A different example is the family of all BSCs, and in this case the maximum likelihood (ML) estimator is simply $\hat Q_{\bar\bxt\bar\byt}(X\neq Y)$.

At this point, we can we can use the definition of the competitive criterion $K_n(\calC,\calR^e)$ in \eqref{Kn(R,C) definition}, and define $\bar K_n^e \triangleq \bE\ppp{K_n(\calC,\calR^e)}$, where the expectation is w.r.t. the first training ensemble defined above. As before, for a given rate $R$ and threshold $T$, we will be interested in the maximal achievable $\xi$ such that 
\begin{align}
\xi^e(R,T,\alpha,\bar P_X) \triangleq \sup\ppp{\xi\in\pp{0,1}:\;\limsup_{n\to\infty}\frac{1}{n}\log \bar K_{n}^e\leq0}.\label{XicriterionTrain}
\end{align}
The above definition sets the stage for a reasonable criterion of optimal training, which includes both the relative training time and the optimal (type of the) training sequence.
In other words, the training fraction $\alpha$ and training type $\bar P_X$ can be optimized to obtain, 
\begin{align}
\xi^e(R,T)\triangleq \max_{0<\alpha<1}\max_{\bar P_X} \xi^e(R,T,\alpha,\bar P_X).\label{optialXiPlugin}
\end{align}
Contrary to the universal decoder considered in the previous subsections, here, we can easily extract $\xi^e(R,T,\alpha,\bar P_X)$, as it appears only in the denominator of \eqref{Kn(R,C) definition}. Indeed, letting $E^e_i(R,T,\theta)$, for $\;i=1,2$, be the error exponents associated with the plug-in decoder $\calR^e$, and the training ensemble defined above, using \eqref{Kn(R,C) definition} and \eqref{XicriterionTrain}, it is easy to verify that \eqref{xi estimate formula0}-\eqref{xi estimate formula}, shown at the top of the next page, hold. 
\begin{figure*}[!t]
\normalsize
\setcounter{MYtempeqncnt}{\value{equation}}
\setcounter{equation}{68} 
\begin{align}
\xi^e(R,T,\alpha,\bar P_X) &= \min_{\theta\in\Theta}\ppp{-\frac{1}{E_{1}(R,T,\theta)}\pp{T+\limsup_{n\to\infty}\frac{1}{n}\bE\ppp{\Gamma_{\theta}({\cal C},{\calR^e})}}}\label{xi estimate formula0} \\
&= \min_{\theta\in\Theta}\ppp{-\frac{1}{E_{1}(R,T,\theta)}\pp{T-\min \ppp { E^e_1(R,T,\theta) + T , E^e_2(R,T,\theta)}}} \\
&= \min_{\theta\in\Theta}\ppp{\frac{1}{E_{1}(R,T,\theta)}\pp{\min \ppp { E^e_1(R,T,\theta), E^e_2(R,T,\theta) -T}}}. \label{xi estimate formula}
\end{align}
\hrulefill
\vspace*{4pt}
\end{figure*}
Similar results can be obtained for the second training ensemble. Note that for the second ensemble $\xi$ should be monotonically increasing with $\beta$, because the more we train the plug-in decoder, the better we compete with the informed decoder. Accordingly, when $\beta\to\infty$ the plug-in decoder actually knows the channel, so the maximal $\xi$ should be trivially one, but $R_{\text{eff}}$ is zero. In between these two extremes, we get the entire spectrum of trade-offs between the maximal $\xi$ and $R_{\text{eff}}$. 

\subsection{Numerical examples}
Consider the setting of Example \ref{exmp:1}, in which a family of BSCs is studied where $\theta\in\Theta$ designates the cross-over probability of the BSC, and $\Theta=\{0,1/100,2/100,\ldots,1\}$. Due to the symmetry of the channels, we take $P_X = \bar P_X = (1/2,1/2)$. The plug-in decoder employs the ML estimator in the initial estimation stage to estimate the unknown crossover probability.

For a given rate $R$ and threshold $T$ we will plot the error exponent achieved for any given $\theta\in\Theta$ by the various decoders. In light of the discussion in Subsection \ref{subsec:unidecoder}, for both the optimal decoder for a known channel and the universal decoder \eqref{universal_decoder}, we will assume that there is no training, i.e., $\alpha=0$ (or, $\beta=0$). From the proof of Theorem \ref{Th:1}, it is evident that the exponents achieved by the universal decoder \eqref{universal_decoder} are given by
\begin{align}
&E_1^u(R,T,\theta) \triangleq \min\left\{\vphantom{\min_{(Q,\tilde{Q})\in\calQ}}\min_{\tilde{Q}\in\calL}D({\tilde{Q}}||P_X\times W_\theta),\right.\nonumber\\
&\ \ \ \ \ \ \ \ \ \ \left.\min_{(Q,\tilde{Q})\in\calQ}\ppp{D({\tilde{Q}}||P_X\times W_\theta)+I(Q)-R}\right\},\label{exponentToeavluate}
\end{align}
and $E_2^u(R,T,\theta)=E_1^u(R,T,\theta)+T$, due to Lemma \ref{lem:2} (see, Appendix \ref{app:1}). To evaluate \eqref{exponentToeavluate}, in every instance of $\xi$ (e.g., \eqref{Gdef}), we substitute $\xi^*(R,T)$ which was already calculated in Example \ref{exmp:1}. Finally, the exponents of the plug-in decoder are given in \eqref{ExponentPlugIn1} and \eqref{ExponentPlugIn2}, for the two ensembles, respectively. In our simulations, we choose $R=T=0.1$, for which $\xi^*(R,T)=0.66$, and we use $\alpha=0.25$, which turns out to be the (approximately) optimal length of the training sequence, for all $\theta\in\Theta$. Fig. \ref{fig:1} compares the error exponents achieved by the various decoders (i.e., optimal decoder for known channel, universal decoder, and plug-in decoder), as a function of $\theta$, using the first ensemble (defined via $\alpha$) for the plug-in decoder. It can be seen that there is a noticeable loss in using the plug-in decoder compared to the universal decoder. Fig. \ref{fig:2} compares the error exponents achieved by the various decoders, as a function of $\theta$, using the second ensemble (defined via $\beta$) for the plug-in decoder, using two values of $\beta$. From this figure, it can be seen that for $\beta=0.32$ the performance of the plug-in decoder are close to the universal decoder, and for $\beta=0.5$ the performance are fairly close to the known channel decoder. Recall, however, that the price in using $\beta=0.32$ and $\beta=0.5$ is an effective rate of $R_{\text{eff}} = 0.76\cdot R$ and $R_{\text{eff}} = (2/3)\cdot R$, respectively.
 
\begin{figure}[!t]
\begin{minipage}[b]{1.0\linewidth}
  \centering
	\centerline{\includegraphics[width=10cm,height = 7cm]{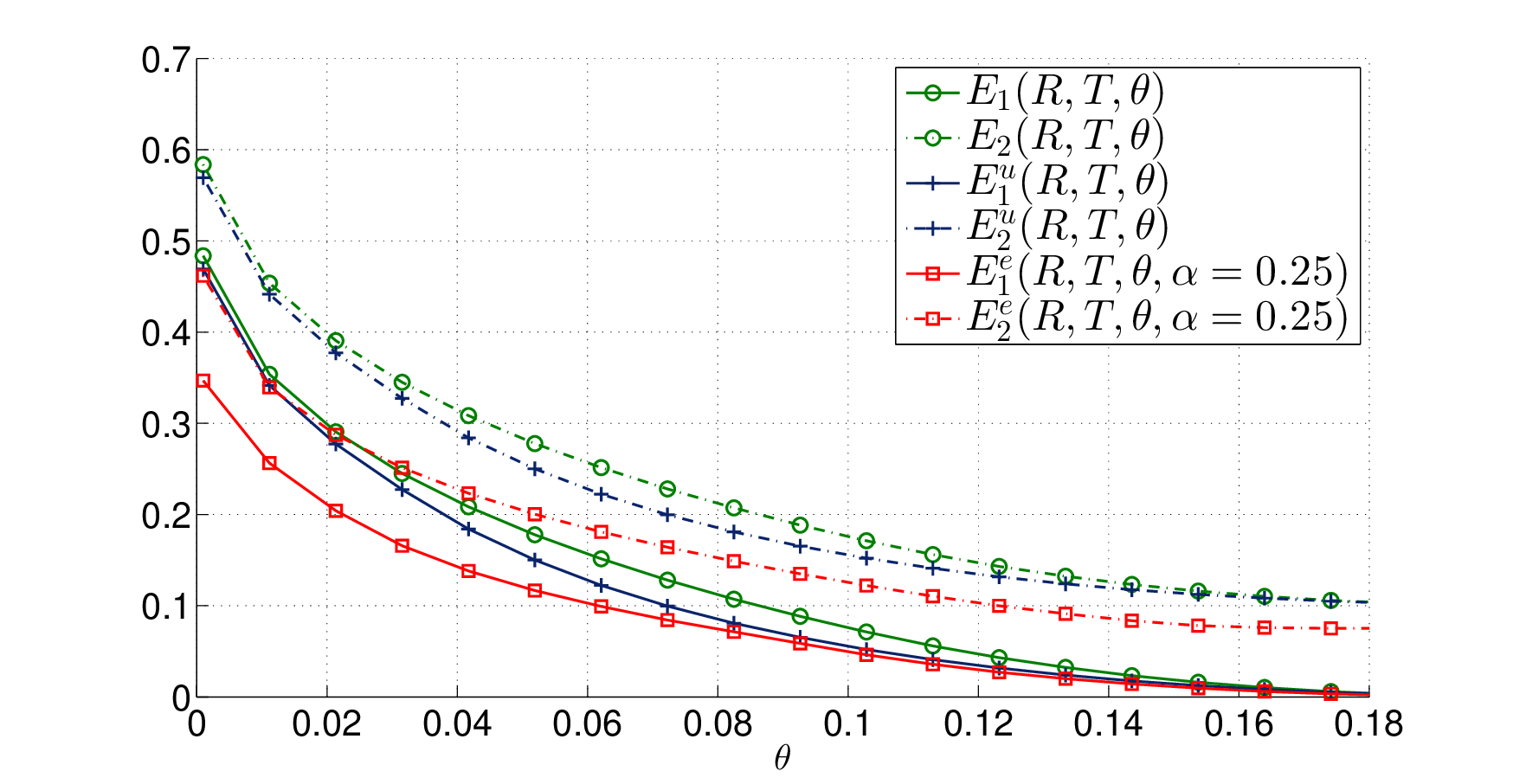}}
	\end{minipage}
\caption{Comparison of the error exponents achieved by the optimal decoder for known channel, universal decoder, and the plug-in decoder, as a function of $\theta$, for $R=T=0.1$, $\xi^*(R,T)=0.66$, and using the first ensemble with $\alpha=0.25$.}
\label{fig:1}
\end{figure}

\begin{figure}[!t]
\begin{minipage}[b]{1.0\linewidth}
  \centering
	\centerline{\includegraphics[width=10cm,height =7cm]{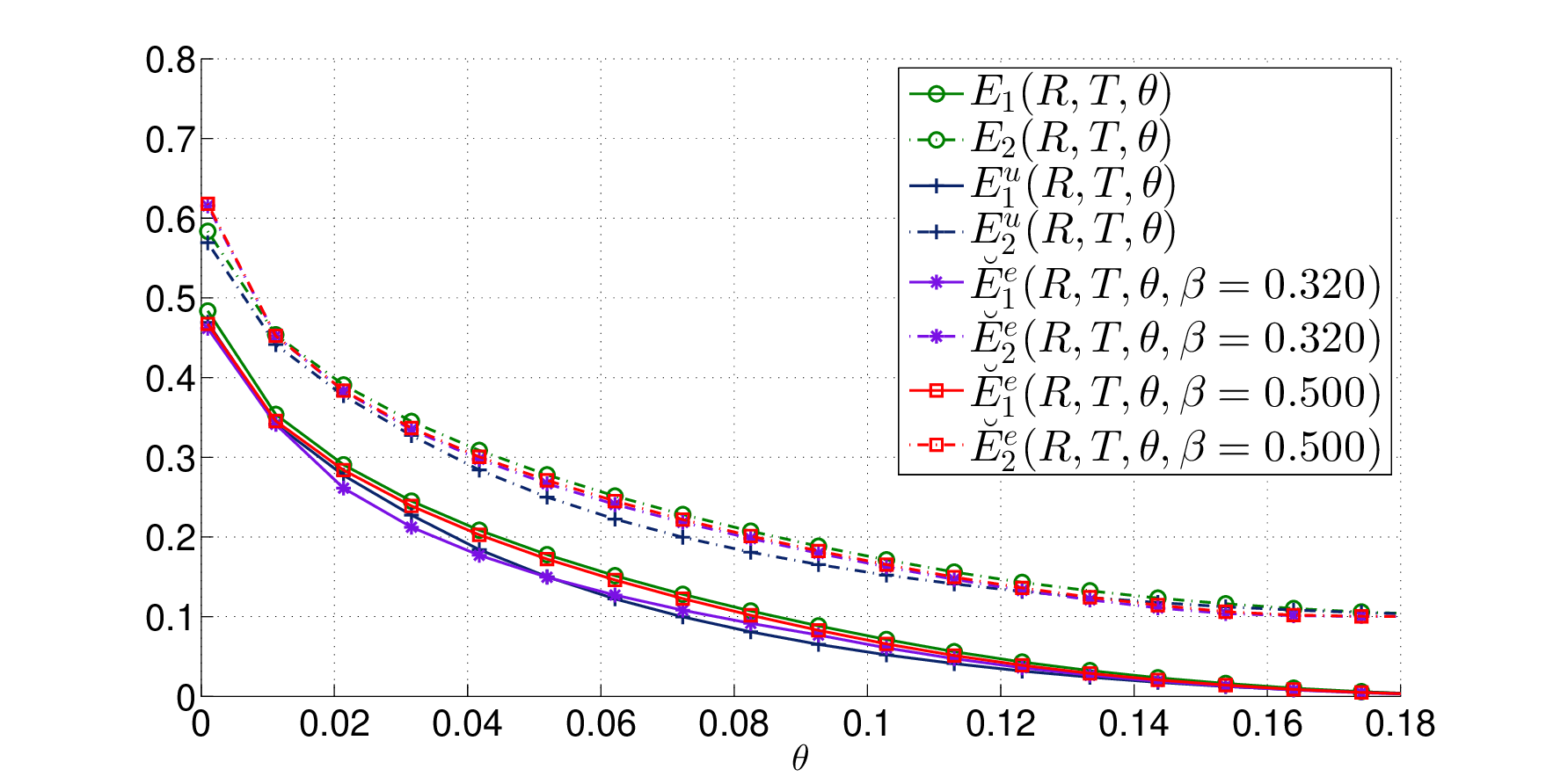}}
	\end{minipage}
\caption{Comparison of the error exponents achieved by the optimal decoder for known channel, universal decoder, and the plug-in decoder, as a function of $\theta$, for $R=T=0.1$, $\xi^*(R,T)=0.66$, and using the second ensemble with $\beta=0.32,0.5$.}
\label{fig:2}
\end{figure}

\begin{rem}
Remarkably, in our numerical calculations we get that $\xi^e(R,T)=0$ (defined in \eqref{optialXiPlugin}), for all $R\geq0$ and $T>0$. This result may be attributed to the fact that our competitive criterion implicitly assumes that the difference between the total error exponent and the undetected error exponent is $T$ (and rightfully, as this is true for both the optimal decoder in the case of a known channel, and for the asymptotically optimal decoder in the case of unknown channel). However, this is not necessarily true for the plug-in decoder, and $E^e_2(R,T,\theta)$ maybe less than $E^e_1(R,T,\theta)+T$, and so the undetected error exponent of the plug-in decoder poorly competes with $\xi\cdot E_1(R,T,\theta)+T$ (recall the definition in \eqref{equminmax}). For this example, no value of $\alpha$ has produced $E_2^e(R,T,\theta)\geq T$ uniformly over $\theta\in\Theta$ and this resulted in the zero values $\xi^e(R,T)$ (recall \eqref{xi estimate formula}). So, even in this relatively simple example, using a plug-in decoder will cause a significant loss in error exponents. In light of this result, a less pessimistic criterion, could be
\begin{align}
\min_\calR\max_{\theta\in\Theta}\frac{\Gamma_\theta\p{\calC,\calR}}{e^{-n\pp{\bar\xi\cdot E_1\p{R,T,\theta}+\bar\xi T}}}, \label{equminmax_alternative}
\end{align}
where now $T$ is also multiplied by $\bar\xi$. A fraction $\xi$ achieved under this criterion implies that the plug-in decoder simultaneously achieves exponents of ${E}_1^e(R,T,\theta)\geq \bar\xi\cdot E_1(R,T,\theta)$ and ${E}_2^e(R,T,\theta)\geq \bar\xi\cdot E_2(R,T,\theta)=\xi\cdot \p{E_1(R,T,\theta)+T}$  for all $\theta\in\Theta$. The analysis of the maximal achievable fraction $\bar\xi(R,T)$ that pertains to \eqref{equminmax_alternative} is the same as for $\xi^e(R,T)$ under the original criterion. Of course, this alternative criterion will lead to different numerical values for $\xi^e(R,T)$ (specifically, positive values for $T>0$). 
\end{rem}

\section{Proofs}\label{sec:proofs}
In the following, for simplicity of notations, we omit the dependency of the various quantities on $R$, $T$, and $\xi$, as they remain constants along the proofs, e.g., $\Omega(R,T,\xi,Q,\tilde{Q})$ will be replaced with $\Omega(Q,\tilde{Q})$.
\begin{proof}[Proof of Theorem \ref{Th:1}]
We analyze the total error term, following the steps of \cite[Section V]{exact_erasure}. 
As was mentioned earlier, we want to assess the (exact) exponential behavior of $\bE\left[K_{n}({\cal C},\hat{{\cal R}})\right]$. In \cite[Theorem 2]{universal_minimax_erasure}, an upper bound was derived on this quantity, so here we seek a tight lower bound.
Let $\Theta_{n}$ denote the set of values of $\theta$ that achieve the maximum at the right-hand side of \eqref{ def of f} for some $\bx\in{\cal X}^{n},\by\in{\cal Y}^{n}$. Note that the elements of $\Theta_n$ depend on $\bx$ and $\by$ only through their joint type, and whence, we have that $\left|\Theta_{n}\right|\leq(n+1)^{|{\cal X}|\cdot|{\cal Y}|-1}$, i.e. the size of $\Theta_{n}$ is a polynomial function of $n$. Now, 
\begin{align}
&\bE\left[K_{n}({\cal C},\hat{{\cal R}})\right] =  \bE\left\{ \max_{\theta\in\Theta}\frac{\Gamma_{\theta}({\cal C},\hat{{\cal R}})}{e^{-n\left(\xi E_{1}(\theta)+T\right)}}\right\}\nonumber\\
 & \geq  \bE\left\{ \max_{\theta\in\Theta_{n}}\frac{\Gamma_{\theta}({\cal C},\hat{{\cal R}})}{e^{-n\left(\xi E_{1}(\theta)+T\right)}}\right\} \nonumber\\
 & \overset{{\scriptstyle (a)}}{\doteq}  \bE\left\{ \sum_{\theta\in\Theta_{n}}\frac{\Gamma_{\theta}({\cal C},\hat{{\cal R}})}{e^{-n\left(\xi E_{1}(\theta)+T\right)}}\right\} \nonumber\\
 & \overset{{\scriptstyle (b)}}{=}  \bE\left\{ \sum_{\theta\in\Theta_{n}}\frac{\frac{1}{M}\sum_{m=1}^{M}\sum_{\byt\in\hat{{\cal R}}_{m}}\sum_{m'\neq m}W_{\theta}(\by|\bX_{m'})}{e^{-n\left(\xi E_{1}(\theta)+T\right)}}\right.\nonumber\\
&\left.\ \ \ \ \ \ \ \ \ \ \ \ \ \ \ \ \ +\frac{\frac{1}{M}\sum_{m=1}^{M}\sum_{\byt\in\hat{{\cal R}}_{m}^{c}}e^{-nT}W_{\theta}(\by|\bX_{m})}{e^{-n\left(\xi E_{1}(\theta)+T\right)}}\right\} \nonumber\\
 & =  \bE\left\{ \frac{1}{M}\sum_{m=1}^{M}\sum_{\byt\in\hat{{\cal R}}_{m}}\sum_{m'\neq m}\sum_{\theta\in\Theta_{n}}e^{n\left(\xi E_{1}(\theta)+T\right)}W_{\theta}(\by|\bX_{m'})\right\} \nonumber\\
 &  \ \ \ +\bE\left\{ \frac{1}{M}\sum_{m=1}^{M}\sum_{\byt\in\hat{{\cal R}}_{m}^{c}}\sum_{\theta\in\Theta_{n}}e^{n\xi E_{1}(\theta)}W_{\theta}(\by|\bX_{m})\right\} \nonumber\\
 & \overset{{\scriptstyle (c)}}{\doteq}  \bE\left\{ \frac{1}{M}\sum_{m=1}^{M}\sum_{\byt\in\hat{{\cal R}}_{m}}\sum_{m'\neq m}\max_{\theta\in\Theta_{n}}e^{n\left(\xi E_{1}(\theta)+T\right)}W_{\theta}(\by|\bX_{m'})\right\} \nonumber\\
 &  \ \ \ +\bE\left\{ \frac{1}{M}\sum_{m=1}^{M}\sum_{\byt\in\hat{{\cal R}}_{m}^{c}}\max_{\theta\in\Theta_{n}}e^{n\xi E_{1}(\theta)}W_{\theta}(\by|\bX_{m})\right\} \nonumber\\
 & =  \bE\left\{ \frac{1}{M}\sum_{m=1}^{M}\sum_{\byt\in\hat{{\cal R}}_{m}}\sum_{m'\neq m}f(\bX_{m'},\by)\right\}\nonumber\\
&\ \ \ \ \ \  +\bE\left\{ \frac{1}{M}\sum_{m=1}^{M}\sum_{\byt\in\hat{{\cal R}}_{m}^{c}}e^{-nT}f(\bX_{m},\by)\right\} \label{LastPass}
\end{align}
where in $(a)$ and $(c)$ we have used the fact that the size of $\Theta_n$ is polynomial, and thus can be absorbed in the $e^{nT}$ factor (see, \cite[pp. 5, footnote 2]{exact_erasure}), and (b) follows from \eqref{gammaTh}. As was shown in \cite[eq. after (A.1)]{universal_minimax_erasure}, the lower bound in \eqref{LastPass} is, in fact, also an upper bound on $\bE\left[K_{n}({\cal C},\hat{{\cal R}})\right]$. Therefore, in the exponential scale, nothing was lost due to the above bounding, and we essentially have that
\begin{align}
\bE\left[K_{n}({\cal C},\hat{{\cal R}})\right]&\doteq \bE\left\{ \frac{1}{M}\sum_{m=1}^{M}\sum_{\byt\in\hat{{\cal R}}_{m}}\sum_{m'\neq m}f(\bX_{m'},\by)\right\} \nonumber\\
&\ \ +\bE\left\{ \frac{1}{M}\sum_{m=1}^{M}\sum_{\byt\in\hat{{\cal R}}_{m}^{c}}e^{-nT}f(\bX_{m},\by)\right\}.\label{LastPass2}
\end{align}
Contrary to the proof technique used in \cite{universal_minimax_erasure} to assess the exponential behavior of \eqref{LastPass2}, where Chernoff and Jensen bounds were invoked, here, we will evaluate the \emph{exact} exponential scale of the two terms on the right hand side of \eqref{LastPass2}. It can be noticed that the first expression is related to undetected errors (or average number of incorrect codewords on the list), and the second one is related to the total error (erasures and undetected errors). For brevity, we define
\begin{align}
A_{1}\triangleq e^{-nT}\cdot\bE\left\{ \frac{1}{M}\sum_{m=1}^{M}\sum_{\byt\in\hat{{\cal R}}_{m}^{c}}f(\bX_{m},\by)\right\},\label{A1term}
\end{align}
and 
\begin{align}
A_{2} & \triangleq  \bE\left\{ \frac{1}{M}\sum_{m=1}^{M}\sum_{\byt\in\hat{{\cal R}}_{m}}\sum_{m'\neq m}f(\bX_{m'},\by)\right\},\label{A2term}
\end{align}
and so
\[
\bE\left[K_{n}({\cal C},\hat{{\cal R}})\right]\doteq A_{1}+A_{2}.
\]
As was mentioned before, we would like to analyze the exponential rate of \eqref{LastPass2}, or, equivalently, of \eqref{A1term} and \eqref{A2term}. Now, note that,
\begin{align}
&\lim_{n\to\infty}\frac{1}{n}\log\bE\left[K_{n}({\cal C},\hat{{\cal R}})\right]\nonumber\\
&\ \ \ \ \ \ \ \ \ =\max\left\{ \lim_{n\to\infty}\frac{1}{n}\log A_{1},\lim_{n\to\infty}\frac{1}{n}\log A_{2}\right\},
\end{align}
whenever all the limits exist. Then, a fraction $\xi$ is \emph{achievable} if both $n^{-1}\log A_{1}$ and $n^{-1}\log A_{2}$ converge to a non-positive constant as $n\to\infty$. 
Let us begin with the evaluation of $A_1$. Continuing from \eqref{A1term}, we get \eqref{lastTerm0}-\eqref{lastTerm}, shown at the top of the next page, 
\begin{figure*}[!t]
\normalsize
\setcounter{MYtempeqncnt}{\value{equation}}
\setcounter{equation}{79} 
\begin{align}
A_{1} &=  e^{-nT}\bE\left\{ \frac{1}{M}\sum_{m=1}^{M}\sum_{\byt}f(\bX_{m},\by)\cdot\calI\{\by\in\hat{{\cal R}}_{m}^{c}\}\right\}\label{lastTerm0}\\
 & \overset{{\scriptstyle (a)}}{=}  e^{-nT}\bE\left\{\left. \sum_{\byt}f(\bX_{m},\by)\cdot\calI\{\by\in\hat{{\cal R}}_{m}^{c}\}\right|m\mbox{th message transmitted}\right\} \\
 & =  e^{-nT}\sum_{\byt}\bE\left\{\left. f(\bX_{m},\by)\cdot\calI\{\by\in\hat{{\cal R}}_{m}^{c}\}\right|m\mbox{th message transmitted}\right\} \\
 & =  e^{-nT}\sum_{\bxt_{m}}P_{X}(\bX_{m}=\bx_{m})\sum_{\byt}\bE\left\{\left. f(\bX_{m},\by)\cdot\calI\{\by\in\hat{{\cal R}}_{m}^{c}\}\right|\bX_{m}=\bx_{m},m\mbox{th message transmitted}\right\} \\
 & =  e^{-nT}\sum_{\bxt_{m}}P_{X}(\bX_{m}=\bx_{m})\sum_{\byt}f(\bx_{m},\by)\cdot\Pr\left\{\left. \by\in\hat{{\cal R}}_{m}^{c}\right|\bX_{m}=\bx_{m},m\mbox{th message transmitted}\right\}\label{lastTerm}
\end{align}
\hrulefill
\vspace*{4pt}
\end{figure*}
where (a) follows from the symmetry of the random coding mechanism, and the probability in the last equation is over the random choice of $\ppp{\bX_{m'}}_{m'\neq m}$, which determines $\hat{{\cal R}}_{m}$. Now, if $Q$ is the joint empirical probability distribution (defined on $\calX\times\calY$) of $\bx_{m'}$ and $\by$, then,
\begin{align}
&f(\bx_{m'},\by) =  \max_{\theta\in\Theta}\left\{ e^{n\left(\xi E_{1}(\theta)+T\right)}W_{\theta}(\by|\bx_{m'})\right\} \\
 & =  \max_{\theta\in\Theta}\left\{ e^{n\left(\xi E_{1}(\theta)+T\right)}e^{n\bE_{Q}\log W_\theta(Y\vert X)}\right\} \\
 & =  \exp\left[n\cdot\max_{\theta\in\Theta}\left\{ \left(\xi E_{1}(\theta)+T\right)+\bE_{Q}\log W_\theta(Y\vert X)\right\} \right]\\
 & =  \exp\left[n\cdot G(Q)\right],
\end{align}
where
\begin{align}
G(Q) \triangleq \max_{\theta\in\Theta} \ppp{\xi  E_1\p{\theta}+T+\bE_{\tilde{Q}}\log W_\theta(Y\vert X)}.\label{QdefProof}
\end{align}
Next, we shall focus on the latter probability in \eqref{lastTerm}. For a given $\bx_m$ and $\by$, let $\tilde{Q} = \hat{P}_{\bxt\byt}$, and let $N_{\byt}(Q)$ denote the number of codewords (excluding $\bx_{m}$) whose joint empirical probability distribution with a given $\by$ is $Q$. Accordingly, we have that
\begin{align}
&\Pr\left\{ \by\in\hat{{\cal R}}_{m}^{c}|\bX_{m}=\bx_{m},\bY=\by\right\}\nonumber\\
 & =  \Pr\left\{ \sum_{m'\neq m}f(\bX_{m'},\by)\geq f(\bx_{m},\by)e^{-nT}\right\}\nonumber\\
 & =  \Pr\left\{ \sum_{Q}N_{\byt}(Q)\exp\left[n\cdot G(Q)\right]\geq\exp\pp{n\cdot G(\tilde{Q})}e^{-nT}\right\}\nonumber\\
 & \doteq  \Pr\left\{ \max_{Q}N_{\byt}(Q)\exp\left[n\cdot G(Q)\right]\geq\exp\pp{n\cdot G(\tilde{Q})}e^{-nT}\right\}\nonumber\\
 & =  \Pr\left\{ \bigcup_{Q} N_{\byt}(Q)\exp\left[n\cdot G(Q)\right]\geq\exp\pp{n\cdot G(\tilde{Q})}e^{-nT} \right\}\nonumber\\
 & \doteq  \sum_{Q}\Pr\left\{ N_{\byt}(Q)\exp\left[n\cdot G(Q)\right]\geq\exp\pp{n\cdot G(\tilde{Q})}e^{-nT}\right\}\nonumber\\
 & \doteq  \max_{Q}\Pr\left\{ N_{\byt}(Q)\exp\left[n\cdot G(Q)\right]\geq\exp\pp{n\cdot G(\tilde{Q})}e^{-nT}\right\}\nonumber\\
 & = \max_{Q\in\calS(\hat{P}_{\byt})}\Pr\left\{ N_{\byt}(Q)\geq\exp\pp{n\cdot\Omega(Q,\tilde{Q})}\right\}\label{maxMes}
\end{align}
where
\begin{align}
\Omega(Q,\tilde{Q}) \triangleq G(R,T,\xi,\tilde{Q})-G(R,T,\xi,Q)-T
\end{align}
and for a given $\breve{Q}_Y$, $\calS(\breve{Q}_Y)\triangleq\{Q:\;Q_Y = \breve{Q}_Y\}$. The asymptotic analysis of the probability in \eqref{maxMes} was carried out in \cite[Section V]{exact_erasure} for any given $\Omega$, and it is not different here. The result relies on the exponential decay of the probability that the joint type of a given $\by$ with a randomly chosen $\bx_{m'}$ is $Q$, namely 
\[
p\triangleq\Pr\ppp{\hat{P}_{\bXt_{m'},\byt}=Q}. \label{probability of joint type}
\]
Under the assumed random coding ensemble, a simple application of the method of types reveals that \cite{csiszar2011information}
\[
p\doteq\exp\ppp{-nI(Q)}.\label{Mutu}
\]
Next, standard large deviations arguments (cf. \cite[Section V]{exact_erasure}) reveal that for $Q\in\calS(\hat{P}_{\byt})$, we have \eqref{largeDe}, shown at the top of the next page,
\begin{figure*}[!t]
\normalsize
\setcounter{MYtempeqncnt}{\value{equation}}
\setcounter{equation}{93} 
\begin{align}
&\Pr\left\{ N_{\by}(Q)\geq e^{n\Omega(Q,\tilde{Q})}\right\} \doteq 
\begin{cases}
\exp\left\{ -n\left|I(Q)-R\right|^{+}\right\} \ & \Omega(Q,\tilde{Q})\leq0\\
1 \ &0<\Omega(Q,\tilde{Q})\leq R-I(Q)\\
0 \ & \Omega(Q,\tilde{Q})>R-I(Q)
\end{cases},\label{largeDe}
\end{align}
\hrulefill
\vspace*{4pt}
\end{figure*}
where by $a_n\doteq 0 $ we mean that $a_n$ decreases to $0$ super-exponentially fast. Define $U(\tilde{Q})$ in \eqref{Udef}.
\begin{figure*}[!t]
\normalsize
\setcounter{MYtempeqncnt}{\value{equation}}
\setcounter{equation}{94} 
\begin{align}
U(\tilde{Q})\triangleq \max_{Q\in\calS(\tilde{Q}_Y)}\begin{cases}
\exp\left[ -n(I(Q)-R)\right] \ & \Omega(Q,\tilde{Q})\leq0,\;I(Q)>R\\
1 \ &I(Q)\leq R,\;\Omega(Q,\tilde{Q})\leq R-I(Q)\\
0 \ &\text{otherwise}
\end{cases}.\label{Udef}
\end{align}
\hrulefill
\vspace*{4pt}
\end{figure*}
Thus, substituting \eqref{largeDe} in \eqref{maxMes} and then in \eqref{lastTerm}, we obtain, using the method of types,
\begin{align}
A_{1} & \doteq  e^{-nT}\sum_{\bx_{m}}P(\bX_{m}=\bx_{m})\sum_{\by}f(\bx_{m},\by)\cdot U(\tilde{Q})\\
 & \doteq  e^{-nT}\max_{\tilde{Q}}\exp\pp{nH_{Y\vert X}(\tilde{Q})}\exp\left[n G(\tilde{Q})\right] U(\tilde{Q}).
\end{align}
Note that the condition:
\[
\Omega(Q,\tilde{Q})\leq R-I(Q)
\]
in \eqref{Udef} is equivalent to
\[
G(\tilde{Q})\leq G(Q)-I(Q)+R+T.
\]
Thus, we obtain that the exponent of $A_{1}$ is given by
\[
\lim_{n\to\infty}\frac{1}{n}\log A_{1}= -T-\min\ppp{\tilde{E}_{a}(R,T,\xi),\tilde{E}_{b}(R,T,\xi)},
\]
in which
\[
\tilde{E}_{a}(R,T,\xi)\triangleq\min_{(Q,\tilde{Q})\in\calQ}\left[-H_{Y\vert X}(\tilde{Q})-G(\tilde{Q})+I(Q)-R\right]\label{Ea}
\]
where
\begin{align}
\calQ\triangleq \ppp{(Q,\tilde{Q})\in\calD:\;I(Q)\geq R,\ \Omega(R,T,\xi,Q,\tilde{Q})\leq0},
\end{align}
and 
\[
\tilde{E}_{b}(R,T,\xi)\triangleq\min_{\tilde{Q}\in{\cal L}}\left[-H_{Y\vert X}(\tilde{Q})-G(\tilde{Q})\right]
\]
where
\begin{align}
\calL&\triangleq\left\{\vphantom{\max_{Q:(Q,\tilde{Q})\in\calD,\;I(Q)\leq R}}\tilde{Q}:G(R,T,\xi,\tilde{Q})\leq R+T\right.\nonumber\\
&\left.\ \ +\max_{Q:(Q,\tilde{Q})\in\calD,\;I(Q)\leq R}\pp{G(R,T,\xi,Q)-I(Q)}\right\}.
\end{align}
Now, we want to find the maximal $\xi$ for which 
\begin{align}
-T-\tilde{E}_{a}(R,T,\xi)\leq0,
\end{align}
\begin{align}
-T-\tilde{E}_{b}(R,T,\xi)\leq0.
\end{align}
For $\tilde{E}_{a}(R,T,\xi)$, substituting $G(Q)$, given in \eqref{QdefProof}, in \eqref{Ea}, we obtain \eqref{maxtht0}-\eqref{maxtht}, shown at the top of the page,
\begin{figure*}[!t]
\normalsize
\setcounter{MYtempeqncnt}{\value{equation}}
\setcounter{equation}{106} 
\begin{align}
-\tilde{E}_{a}(R,T,\xi)-T &= \max_{(Q,\tilde{Q})\in\calQ}\pp{H_{Y\vert X}(\tilde{Q})+G(\tilde{Q})-I(Q)+R}-T\label{maxtht0}\\
&= \max_{(Q,\tilde{Q})\in\calQ}\pp{H_{Y\vert X}(\tilde{Q})+\max_\theta \ppp{\xi E_1\p{\theta}+\bE_{\tilde{Q}}\log W_\theta(Y\vert X)}-I(Q)+R}\\
&=\max_\theta\ppp{\xi E_1\p{\theta}+\max_{(Q,\tilde{Q})\in\calQ}\ppp{H_{Y\vert X}(\tilde{Q})-I(Q)+R+\bE_{\tilde{Q}}\log W_\theta(Y\vert X)}}\\
&=\max_\theta\ppp{\xi E_1\p{\theta}-\min_{(Q,\tilde{Q})\in\calQ}\ppp{D(\tilde{Q}||P_X\times W_\theta)+I(Q)-R}},\label{maxtht}
\end{align}
\hrulefill
\vspace*{4pt}
\end{figure*}
which is exactly the condition in \eqref{mainEq}. In a similar manner, one obtains
\begin{align}
&-\tilde{E}_{b}(R,T,\xi)-T \nonumber\\
&\ \ \ \ \ \ \ \ = \max_\theta\ppp{\xi E_1\p{\theta}-\min_{\tilde{Q}\in\calL}\;D(\tilde{Q}||P_X\times W_\theta)},\label{maxtht2}
\end{align}
which is exactly the condition in \eqref{mainEq2}. This concludes the analysis of $A_1$, and we next consider $A_2$. In essence, we can derive the exponential behavior of $A_2$, using similar methods to the derivation of $E_2(R,T,\theta)$ in \cite{exact_erasure}. However, since the resulting exponent $\lim_{n\to\infty}\frac{1}{n}\log A_{2}$ is continuous in $T$, just as $\lim_{n\to\infty}\frac{1}{n}\log A_{1}$, we may invoke the following lemma, which is analogue to Lemma \ref{lem:1}, and is proved in Appendix \ref{app:1}:
\begin{lemma} \label{lem:2}
For all $R$ and $T$:
\begin{align}
\lim_{n\to\infty}\frac{1}{n}\log A_{2} = T+\lim_{n\to\infty}\frac{1}{n}\log A_{1}.
\end{align}
\end{lemma}

Thus, it suffices to asses the exponent of either $A_1$ or $A_2$, and then the other one is immediately obtained. While both $A_1$ and $A_2$ can be analyzed, the analytical formula for the exponent of $A_1$ is more compact, and thus we only presented it. 
\end{proof}

\begin{proof}[Proof of Corollary \ref{cor:1}]
Define the set $\calG(\tilde Q_Y)\triangleq\{Q:\;I(Q)<R,\;Q_Y = \tilde Q_Y\}$. We start from the first condition in Theorem \ref{Th:1}, which is equivalent to requiring that for all $\theta$ and $\tilde Q$
\begin{align}
&\xi E_1\p{R,T,\theta}\leq D({\tilde{Q}}||P_X\times W_\theta)\nonumber\\
&+\min_{Q\in\calG^c(\tilde Q_Y)}\max_{\lambda\geq0}\pp{I(Q)-R+\lambda\cdot\Omega(R,T,\xi,Q,\tilde{Q})}.\label{forallTh}
\end{align}
Letting 
\begin{align}
&\Omega_{\theta_1,\theta_2}(R,T,\xi,Q,\tilde{Q})\triangleq \xi E_1\p{R,T,\theta_1}+\bE_{\tilde{Q}}\pp{\log W_{\theta_1}\p{Y\vert X}}\nonumber\\
&-\pp{\xi E_1\p{R,T,\theta_2}+\bE_{Q}\pp{\log W_{\theta_2}\p{Y\vert X}}} - T,\label{omegath1th2}
\end{align}
we have by definition,
\begin{align}
\Omega(R,T,\xi,Q,\tilde{Q}) = \max_{\theta_1}\min_{\theta_2}\;\Omega_{\theta_1,\theta_2}(R,T,\xi,Q,\tilde{Q}).\label{EquiOme}
\end{align}
Substituting \eqref{EquiOme} in \eqref{forallTh}, we get
\begin{align}
&\xi E_1\p{R,T,\theta}\leq D({\tilde{Q}}||P_X\times W_\theta)\nonumber\\
&\ \ \ \ \ \ \ \ \ \ \ \ \ \ +\min_{Q\in\calG^c(\tilde Q_Y)}\max_{\lambda\geq0}\max_{\theta_1}\min_{\theta_2}\left[\vphantom{\Omega_{\theta_1,\theta_2}(R,T,\xi,Q,\tilde{Q})}I(Q)-R\right.\nonumber\\
&\left.\ \ \ \ \ \ \ \ \ \ \ \ \ \ \ \ \ \ \ \ \ \ \ +\lambda\cdot\Omega_{\theta_1,\theta_2}(R,T,\xi,Q,\tilde{Q})\right],\label{midterm2}
\end{align}
which is equivalent to demanding that for all $Q\in\calG_R^c(\tilde{Q}_Y)$ there exist some $\lambda\geq0$ and $\theta_1\in\Theta$, such that for all $\theta_2\in\Theta$ we have 
\begin{align}
\xi E_1\p{R,T,\theta}&\leq D({\tilde{Q}}||P_X\times W_\theta)+I(Q)-R\nonumber\\
&\ \ \ +\lambda\cdot\Omega_{\theta_1,\theta_2}(R,T,\xi,Q,\tilde{Q}).\label{midterm3}
\end{align} 
Upon substitution of \eqref{omegath1th2} in \eqref{midterm3}, after rearranging the terms, we obtain $\xi\leq\hat{\xi}_1(R,T,Q,\tilde{Q},\theta,\theta_1,\theta_2,\lambda)$, where $\hat{\xi}_1$ is defined in \eqref{hatXi}. Thus, the largest achievable $\xi$ which satisfies the first condition is $\xi_1^*(R,T)$. In the same way, the second condition yields $\xi_2^*(R,T)$, and thus, $\xi^*(R,T) = \min\ppp{\xi_1^*(R,T),\xi_2^*(R,T)}$.
\end{proof}

\begin{proof}[Proof of Corollary \ref{cor:2}]
In the following, we analyze the objective in \eqref{mainEq} for any $\theta$. Starting with the left term, $E_1(\theta)$, note that this is just the expression that was considered in \cite[pp. 6450-6451, eqs. (64)-(73)]{exact_erasure}. For completeness, we present here the main steps in the simplification of this term to the BSC. We start with the analysis of $E_a(R,T)$ given in \eqref{AneliaExp1}. First, note that
\begin{align}
&\bE_{\hat{Q}}\log W_\theta(Y\vert X)-\bE_{Q}\log W_\theta(Y\vert X)\nonumber\\
&\ \ \ \ \ \ \  = \pp{Q\p{X\neq Y}-\hat{Q}\p{X\neq Y}}\beta
\end{align}
where $\beta = \log\pp{\p{1-\theta}/\theta}$. Thus, recalling \eqref{AneliaExp}, $E_1(\theta)$ takes the form
\begin{align}
&\min_{\tilde{Q}}\left\{\vphantom{\abs{\min_{Q\in\hat{\calQ}_{\text{BSC}}(\tilde{Q})}\p{-H_{X\vert Y}(Q)+\log 2-R}}^+}D(\tilde{Q}||P_X\times W_\theta)\right.\nonumber\\
&\ \ \ \ \ \ \ \ \left.+\abs{\min_{Q\in\hat{\calQ}_{\text{BSC}}(\tilde{Q})}\p{-H_{X\vert Y}(Q)+\log 2-R}}^+\right\}
\end{align}
where 
\begin{align}
\hat{\calQ}_{\text{BSC}}(\tilde{Q})&\triangleq\left\{Q:\; Q_Y = \tilde{Q}_Y,\right.\nonumber\\
&\left.\ \ \ \ \ \ \ \ \  Q\p{X\neq Y}\leq\tilde{Q}\p{X\neq Y}+\frac{T}{\beta}\right\}.
\end{align}
Now, note that
\begin{align}
H_{X\vert Y}(Q) = H_{\calI\ppp{X\neq Y}\vert Y}\p{Q}\leq H_{\calI\ppp{X\neq Y}}(Q),
\end{align}
and thus
\begin{align}
&\min_{\tilde{Q}}\left\{\vphantom{\abs{\min_{Q\in\hat{\calQ}_{\text{BSC}}(\tilde{Q})}\p{-H_{\calI\ppp{X\neq Y}}(Q)+\log 2-R}}^+}D(\tilde{Q}||P_X\times W_\theta)\right.\nonumber\\
&\left.\ \ \ \ \ \ \ \ \ \ +\abs{\min_{Q\in\hat{\calQ}_{\text{BSC}}(\tilde{Q})}\p{-H_{X\vert Y}(Q)+\log 2-R}}^+\right\}\nonumber\\
&\geq \min_{\tilde{Q}}\left\{\vphantom{\abs{\min_{Q\in\hat{\calQ}_{\text{BSC}}(\tilde{Q})}\p{-H_{\calI\ppp{X\neq Y}}(Q)+\log 2-R}}^+}D(\tilde{Q}||P_X\times W_\theta)\right.\nonumber\\
&\left.\ \ \ +\abs{\min_{Q\in\hat{\calQ}_{\text{BSC}}(\tilde{Q})}\p{-H_{\calI\ppp{X\neq Y}}(Q)+\log 2-R}}^+\right\}\label{eqeq}\\
& = \min_{\tilde{q}}\ppp{D\p{\tilde{q}||\theta}+\abs{\min_{q\leq\tilde{q}+T/\beta}\p{-h\p{q}+\log 2-R}}^+}
\end{align}
where the last step follows since the minimizing $\tilde{Q}$ is such that $\tilde{Q}_X = P_X$ to obtain minimal $D(\tilde{Q}||P_X\times W_\theta)$, and it is easy to verify using convexity arguments that given $\tilde{Q}(X\neq Y)=\tilde{q}$ the divergence	$D(\tilde{Q}||P_X\times W_\theta)$ is minimized for a symmetric $\tilde{Q}_{Y\vert X}$, namely,
\begin{align}
\tilde{Q}_{Y\vert X}\p{y\vert x} = \begin{cases}
\tilde{q}\ & x=y\\
1-\tilde{q}\ &x\neq y
\end{cases},
\end{align}
for which $D(\tilde{Q}||P_X\times W_\theta) = D\p{\tilde{q}||\theta}$. Finally, it is evident that we have equality in \eqref{eqeq} if we choose
\begin{align}
Q_{Y\vert X}\p{y\vert x} = \begin{cases}
q\ & x=y\\
1-q\ &x\neq y
\end{cases},
\end{align}
and thus it is the minimizer. Next, we observe that $-h\p{q}$ is a decreasing function of $q$ for $q\in\pp{0,1/2}$ and increasing for $q\in\pp{1/2,1}$. Thus,
\begin{align}
&\min_{\tilde{q}}\left\{\vphantom{\abs{\min_{q\leq\tilde{q}+T/\beta}\p{-h\p{q}+\log 2-R}}^+}D\p{\tilde{q}||\theta}\right.\nonumber\\
&\left.\ \ \ \ \ \ \ \ \ \ +\abs{\min_{q\leq\tilde{q}+T/\beta}\p{-h\p{q}+\log 2-R}}^+\right\}\nonumber\\
& = \min_{\tilde{q}}\left\{\vphantom{\abs{\min_{q\leq\tilde{q}+T/\beta}\p{-h\p{q}+\log 2-R}}^+}D\p{\tilde{q}||\theta}\right.\nonumber\\
&\left.\ \ \ \ \ \ \ \ \ \ +\abs{-h\p{\min\ppp{\frac{1}{2},\tilde{q}+\frac{T}{\beta}}}+\log 2-R}^+\right\}\nonumber\\
& = \min_{\tilde{q}}\left\{\vphantom{h\p{\min\ppp{\delta_{GV}\p{R},\tilde{q}+\frac{T}{\beta}}}}D\p{\tilde{q}||\theta}\right.\nonumber\\
&\left.\ \ \ \ \ \ \ \ \ \ -h\p{\min\ppp{\delta_{GV}\p{R},\tilde{q}+\frac{T}{\beta}}}+\log 2-R\right\}\nonumber\\
& = \min_{\tilde{q}\in\pp{\theta,\delta_{GV}\p{R}-T/\beta}}\pp{D\p{\tilde{q}||\theta}-h\p{\tilde{q}+\frac{T}{\beta}}}+\log 2-R\label{aneliaBSC}
\end{align}
where the last step can be easily verified using monotonicity properties of the binary entropy and divergence \cite[p. 6451 after eq. (72)]{exact_erasure}. Now, we analyze $E_b(R,T)$ given in \eqref{AneliaExp2}. Note that there is no conceptual difference between $E_a(R,T)$ and $E_b(R,T)$, and it can be verified that the latter can be written as
\begin{align}
\min_{\tilde{q}\in\hat{\calL}_{\text{BSC}}}\;D\p{\tilde{q}||\theta}
\end{align}
where 
\begin{align}
\hat{\calL}_{\text{BSC}}&\triangleq \left\{\vphantom{\max_{q:\;R\geq\log2-h(q)}}\tilde{q}: -\tilde{q}\cdot\beta\leq R+T\right.\nonumber\\
&\left.\ \ \ \ +\max_{q:\;R\geq\log2-h(q)}\pp{-q\cdot\beta+h(q)-\log2}\right\}.
\end{align}

Next, for any $\theta$, consider the right term in objective of \eqref{mainEq}. Note that the only difference between the left and the right terms in \eqref{mainEq} is just the inner minimization region. Accordingly, the right term takes the form
\begin{align}
&\min_{\tilde{Q}}\left\{D(\tilde{Q}||P_X\times W_\theta)\right.\nonumber\\
&\ \ \ \ \ \ \ \left.+\abs{\min_{Q\in\calQ_{\text{BSC}}(\tilde{Q})}\p{-H_{X\vert Y}(Q)+\log 2-R}}^+\right\}
\end{align}
where 
\begin{align}
&\calQ_{\text{BSC}}(\tilde{Q})\triangleq\left\{Q:\; Q_Y = \tilde{Q}_Y,\right.\nonumber\\
&\max_{\theta'}\ppp{\xi E_1(\theta')-\beta(\theta')\tilde{Q}\p{X\neq Y}+\log\p{1-\theta'}}\nonumber\\
&-\max_{\theta'}\ppp{\xi E_1(\theta')-\beta(\theta'){Q}\p{X\neq Y}+\log\p{1-\theta'}}\nonumber\\
&\left.\vphantom{Q_Y = \tilde{Q}_Y}-T\leq0\right\}.
\end{align}
Let $\tilde{E}_1\p{\theta}\triangleq E_1\p{\theta}+\log(1-\theta)/\xi$. Then, using exactly the same steps as before, we get
\begin{align}
&\min_{\tilde{Q}}\left\{\vphantom{\abs{\min_{Q\in\calQ_{\text{BSC}}(\tilde{Q})}\p{-H_{X\vert Y}\p{Q}+\log 2-R}}^+}D(\tilde{Q}||P_X\times W_\theta)\right.\nonumber\\
&\left.\ \ \ \ \ \ \ \ \ \ +\abs{\min_{Q\in\calQ_{\text{BSC}}(\tilde{Q})}\p{-H_{X\vert Y}\p{Q}+\log 2-R}}^+\right\}\nonumber\\
&\geq \min_{\tilde{Q}}\left\{\vphantom{\abs{\min_{Q\in\calQ_{\text{BSC}}(\tilde{Q})}\p{-H_{X\vert Y}\p{Q}+\log 2-R}}^+}D(\tilde{Q}||P_X\times W_\theta)\right.\nonumber\\
&\left.\ \ \ +\abs{\min_{Q\in\calQ_{\text{BSC}}(\tilde{Q})}\p{-H_{\calI\ppp{X\neq Y}}(Q)+\log 2-R}}^+\right\}\label{eqeq2}\\
& = \min_{\tilde{q}}\ppp{D\p{\tilde{q}||\theta}+\abs{\min_{q\in\tilde{\calQ}_{\text{BSC}}(\tilde{q})}\p{-h\p{q}+\log 2-R}}^+},
\end{align}
and equality can be achieved choosing $Q$ to be symmetric, as before, and
\begin{align}
\tilde{\calQ}_{\text{BSC}}(\tilde{q})&\triangleq\left\{q: \max_{\theta'}\ppp{\xi \tilde{E}_1(\theta')-\beta(\theta')\cdot\tilde{q}}\right.\nonumber\\
&\left.\ \ -\max_{\theta'}\ppp{\xi \tilde{E}_1(\theta')-\beta(\theta')\cdot q}-T\leq0\right\}.\label{ad2}
\end{align}
Next, we simplify the set $\tilde{\calQ}_{\text{BSC}}(\tilde{q})$. The constraint on $q$ in the definition of $\tilde{\calQ}_{\text{BSC}}(\tilde{q})$, is equivalent to demanding that there exist some $\theta'\in\Theta$ such that the following holds
\begin{align}
\beta(\theta')q-\xi \tilde{E}_1(\theta')\leq T-\max_{\theta''}\ppp{\xi \tilde{E}_1(\theta'')-\beta(\theta'')\cdot\tilde{q}},
\end{align}
or equivalently
\begin{align}
\beta(\theta')q\leq \xi \tilde{E}_1(\theta')+T-\max_{\theta''}\ppp{\xi \tilde{E}_1(\theta'')-\beta(\theta'')\cdot\tilde{q}}.
\end{align}
Now, note that $\beta(\theta')\geq0$ if and only if $\theta'\leq1/2$. Accordingly, this means that, in terms of $q$, $\tilde{\calQ}_{\text{BSC}}(\tilde{q})$ is equivalent to $q\leq q_1^*$ or $q\geq q_2^*$, where $q_1^*$ and $q_2^*$ are given in \eqref{q1Def} and \eqref{q2Def}, respectively. 
Consequently,
\begin{align}
&\min_{\tilde{q}}\ppp{D\p{\tilde{q}||\theta}+\abs{\min_{q\in\tilde{\calQ}_{\text{BSC}}(\tilde{q})}\p{-h\p{q}+\log 2-R}}^+}\nonumber\\
&=\min_{\tilde{q}}\ppp{D\p{\tilde{q}||\theta}+\abs{\p{-g\p{q_1^*,q_2^*}+\log 2-R}}^+}
\end{align}
where $g\p{q_1^*,q_2^*}$ is defined in \eqref{gDef}. Finally, we consider the right term in \eqref{mainEq2}. Using the same steps as above we obtain that
\begin{align}
\min_{\tilde{Q}\in\calL}D({\tilde{Q}}||P_X\times W_\theta) = \min_{\tilde{q}\in\calL_{\text{BSC}}}\;D\p{\tilde{q}||\theta}
\end{align}
where $\calL_{\text{BSC}}$ is defined in \eqref{qqww}-\eqref{qqww1}, shown at the top of the next page,
\begin{figure*}[!t]
\normalsize
\setcounter{MYtempeqncnt}{\value{equation}}
\setcounter{equation}{137} 
\begin{align}
\calL_{\text{BSC}}&\triangleq \left\{\tilde{q}: \max_\theta\pp{\xi E_1(\theta)-\tilde{q}\cdot\beta+\log\theta}\leq R+T+\max_{q:\;R\geq\log2-h(q)}\ppp{\max_\theta\pp{\xi E_1(\theta)-q\cdot\beta+\log\theta}+h(q)-\log2}\right\}\label{qqww}\\
&=\left\{\tilde{q}: \max_\theta\pp{\xi E_1(R,T,\theta)-\tilde{q}\cdot\beta(\theta)+\log\theta}\leq R+T\right.\nonumber\\
&\left.\ \ \ \ \ \ \ \ \ \ \ \ \ \ \ +\max_\theta\pp{\xi E_1(R,T,\theta)-\max\ppp{\theta,\delta_{\text{GV}}(R)}\cdot\beta(\theta)+\log\theta+h(\max\ppp{\theta,\delta_{\text{GV}}(R)})-\log2}\right\}\label{qqww1}
\end{align}
\hrulefill
\vspace*{4pt}
\end{figure*}
where the last step in \eqref{qqww1} follows from the fact that the maximizer $q$ in the optimization problem in \eqref{qqww} is given by $\max\ppp{\theta,\delta_{\text{GV}}(R)}$.  
\end{proof}

\appendices
\numberwithin{equation}{section}
\section{Proof of Lemmas \ref{lem:1} and \ref{lem:2}}
\label{app:1}

We begin with the proof of Lemma \ref{lem:1}. For the sake of this proof, we will explicitly designate the dependence on $T$, and denote the decoder in \eqref{optDD1}-\eqref{optDD2}, with parameter $T$, by ${\cal R}^{*}(T)$. Similarly, we will denote the value of \eqref{gammaTh0} as $\Gamma({\cal C},{\cal R},T)$. As we have mentioned, the decoder minimizing $\Gamma({\cal C},{\cal R},T)$ can be easily seen to be given by ${\cal R}^{*}(T)$. Now, assume conversely, that the exponents associated with $\mathbb{E}[\Gamma({\cal C},{\cal R}^{*}(T),T)]$
satisfy 
\begin{align}
E_{2}(R,T) & <T+E_{1}(R,T).\label{opps}
\end{align}
The opposite case, where the inequality in \eqref{opps} is reversed, can be handled analogously. Accordingly, this means that in the exponential scale, we have 
\begin{align}
\mathbb{E}[\Gamma({\cal C},{\cal R}^{*}(T),T)]\doteq e^{-nE_{2}(R,T)}.
\end{align}
Now, it is evident that $E_{1}(R,T)$ is a monotonically decreasing function of $T$ (allowing more erasures increases $\overline{\Pr}\ppp{\calE_1}$), and $E_{2}(R,T)$ is a monotonically increasing function of $T$ (allowing more erasures decreases $\overline{\Pr}\ppp{\calE_2}$) \cite{exact_erasure}. Now, due to the fact that $E_{1}(R,T)$ and $E_{2}(R,T)$ are continuous functions of $T$ \cite[eqs. (23) and (31)]{exact_erasure}, without loss of essential generality, there exists $\epsilon>0$ and $\delta_{1}\geq0,\delta_{2}>0$ such that
\begin{align}
E_{1}(R,T+\epsilon) & =E_{1}(R,T)-\delta_{1}
\end{align}
and
\begin{align}
E_{2}(R,T+\epsilon) & =E_{2}(R,T)+\delta_{2}
\end{align}
yet
\begin{align}
E_{2}(R,T+\epsilon) & <T+E_{1}(R,T+\epsilon).
\end{align}
Note that since it is not guaranteed that $E_{1}(R,T)$ or $E_{2}(R,T)$ are strictly monotonic, as it might be the case that $\delta_2=0$ too, i.e., regions of plateau. Accordingly, there are several cases to consider. First, if just $E_1(R,T)$ is within a plateau region, then the above arguments remain the same since $\delta_1=0$ but $\delta_2>0$. Secondly, if just $E_2(R,T)$ is within a plateau region, then we claim that this contradicts the optimality of Forney's decoder. Indeed, in this case, if we increase $T$ by some small $\epsilon>0$ (such that $E_2(R,T+\epsilon)$ is within the plateau), we obtain a decoder with exponents $E_2(R,T+\epsilon)= E_2(R,T)$ and $E_1(R,T+\epsilon)<E_1(R,T)$, and yet, due to continuity, $E_2(R,T)<E_1(R,T+\epsilon)$. Thus, we obtained that the optimal decoder ${\cal R}^{*}(T+\epsilon)$ has the same performance as ${\cal R}^{*}(T)$, in terms of $\mathbb{E}[\Gamma({\cal C},{\cal R}^{*}(T),T)]$, but with worse $\overline{\Pr}\ppp{\calE_1}$, which means not the best trade-off between $\overline{\Pr}\ppp{\calE_1}$ and $\overline{\Pr}\ppp{\calE_2}$, and thus contradicting the optimality of Forney's decoder at $T+\epsilon$. Finally, if both exponents are within a region of plateau, we can simply vary $T$ until we leave this region, and thus we can assume that $\delta_2>0$.
To conclude, we obtained that 
\begin{align}
\mathbb{E}[\Gamma({\cal C},{\cal R}(T+\epsilon),T)]&\doteq e^{-nE_{2}(R,T+\epsilon)}\\&\stackrel{.}{<}e^{-nE_{2}(R,T)}\\
&\doteq\mathbb{E}[\Gamma({\cal C},{\cal R}(T),T)]
\end{align}
which contradicts the property that ${\cal R}^{*}(T)$ is the minimizer of $\Gamma({\cal C},{\cal R},T)$. 

The proof of Lemma \ref{lem:2} follows the same steps as above. Indeed, the Lagrangian associated with the universal erasure decoder (see, \eqref{LastPass2}), has a similar structure to the Lagrangian associated with the optimal (known channel) decoder (see, \eqref{gammaTh}). As was mentioned in the proof of Theorem \ref{Th:1}, the exponents of $A_1$ and $A_2$ are both continuous. So, just as the difference between the exponents of $\overline{\Pr}\ppp{\calE_1}$ and $\overline{\Pr}\ppp{\calE_2}$ is $T$, the difference between the exponents of $A_1$ and $A_2$ is also $T$.

\section{Universal Decoder With Training}
\label{app:2}
To present the achieved fraction for the ensemble which includes training, we need to slightly generalize the definitions preceding Theorem \ref{Th:1}. The definitions of $G(R,T,\xi,\tilde Q)$ in \eqref{Gdef} and $\Omega(R,T,\xi,Q,\tilde Q)$ in \eqref{Odef} remain exactly the same. Define,
\begin{align}
J(Q,\bar Q)\triangleq (1-\alpha) \cdot I\p{\frac{Q-\alpha \bar Q}{1-\alpha}} \label{Jdef}.
\end{align}
For a given joint type $\bar Q$, we replace the definition of $\cal{Q}$ in \eqref{Qdef} with
\begin{align}
\calQ(\bar Q)&\triangleq \left\{(Q,\tilde{Q})\in\calD(\bar Q):\;J(Q,\bar Q)\geq R,\right.\nonumber\\
&\left.\ \ \ \ \ \ \ \ \ \ \ \ \ \ \ \ \ \ \ \ \ \ \ \ \ \ \Omega(R,T,\xi,Q,\tilde{Q})\leq0\right\},\label{Qdef_training}
\end{align}
and replace the definition of $\cal{L}$ in \eqref{Ldef} with
\begin{align}
&\calL(\bar Q)\triangleq\left\{ \tilde{Q}:G(R,T,\xi,\tilde{Q})\leq R+T\right.\nonumber\\
&\left.+\max_{Q:(Q,\tilde{Q})\in\calD(\bar Q),\;J(Q,\bar Q)\leq R}\pp{G(R,T,\xi,Q)-J(Q,\bar Q)}\right\},\label{Ldef_training}
\end{align}
where $\calD$ defined in \eqref{calD} is replaced by
\begin{align}
&\calD(\bar Q)\triangleq\left\{(Q,\tilde{Q}):\;Q_X = \tilde{Q}_X = P_X,\ Q_Y = \tilde{Q}_Y,\right.\nonumber\\
&\left.\ \ \ \ \ \ \ \ \ \ \ \ \ \ \frac{Q-\alpha \bar Q}{1-\alpha} \text{ is a probability distribution}\right\}.\label{calD_training}
\end{align}
Finally, define
\begin{align}
\Delta_{\theta}(Q,\bar Q) &\triangleq  \alpha\cdot D\p{\bar Q||\bar P_X\times W_\theta}\nonumber\\
&\ \ \ +(1-\alpha)\cdot D\p{\left.\left.\frac{ Q-\alpha \bar Q}{1-\alpha}\right\vert\right\vert P_X\times W_\theta}.
\end{align}
\begin{theorem}\label{Th:2}
Consider the ensemble defined above with types $\bar P_X$ and $P_X$, and a fixed $\alpha\in[0,1)$. Then, $\xi^*(R,T,\alpha,\bar P_X)$, defined in \eqref{xi star definition}, is equal to the largest number $\xi$ that simultaneously satisfies:
\begin{align}
&\max_{\theta\in\Theta}\left\{\vphantom{\min_{\bar Q:\; \bar Q_X = \bar P_X}\min_{(Q,\tilde{Q})\in\calQ(\bar Q)}}\xi E_1\p{R,T,\theta}-\right.\nonumber\\
&\left.\min_{\bar Q:\; \bar Q_X = \bar P_X}\min_{(Q,\tilde{Q})\in\calQ(\bar Q)}\ppp{\Delta_{\theta}(\tilde Q,\bar Q)+J(Q,\bar Q)-R}\right\}\leq0,\label{mainEq_training}
\end{align}
and
\begin{align}
\max_{\theta\in\Theta}\ppp{\xi E_1\p{R,T,\theta}-\min_{\bar Q:\; \bar Q_X = \bar P_X}\min_{\tilde{Q}\in\calL(\bar Q)}\Delta_{\theta}(\tilde Q,\bar Q)}\leq0.\label{mainEq2_training}
\end{align}
\end{theorem}

Choosing a strictly positive $\alpha$ has the potential to increase $\xi^*(R,T,0,\bar P_X)$. However, the behavior of $\xi^*(R,T,\alpha,\bar P_X)$ as a function of $\alpha$, is typically not monotonic. Indeed, as was mentioned before, on the one hand, as $\alpha$ increases, the decoder has better knowledge of the channel, even if it does not estimate it explicit. On the other hand, the number of available symbols $(1-\alpha)n$ that are used to distinguish the $M=\left\lceil e^{ nR}\right\rceil$ codewords from one another decreases\footnote{Note that the blocklength which is used to gauge the rate is still $n$.}. Thus, we expect that, in general, $\xi^*(R,T,\alpha,\bar P_X)$ will be maximized by some $\alpha^*\in(0,1)$. In addition, the type of the training part $\bar P_X$ may also be optimized. Evidently, Theorem \ref{Th:2} sets the stage for a reasonable criterion of optimal training, which includes both the relative training time and the optimal (type of the) training sequence. Similarly to Corollary \ref{cor:1}, one can derive a formula for $\xi^*(R,T,\alpha,\bar P_X)$, and then, a reasonable objective would be to optimize $\xi^*(R,T,\alpha,\bar P_X)$ over both $\alpha$ and $\bar{P}_X$.
%

\begin{proof}[Proof of Theorem \ref{Th:2}] 
The proof follows the same lines of the proof Theorem \ref{Th:1} so we mainly highlight the differences. We will represent the joint type of the training sequence $\bar{\bx}$ and the training part of $\by$ by $\bar Q$. Also, for a given $\by$ we will denote by $\bar \by$ the first $\alpha\cdot n$ symbols of $\by$ (i.e. the output symbols for the training sequnece). We continue \eqref{lastTerm} as shown in \eqref{lastTerm_training0}-\eqref{lastTerm_training}, presented at the top of the next page.
\begin{figure*}[!t]
\normalsize
\setcounter{MYtempeqncnt}{\value{equation}}
\setcounter{equation}{7} 
\begin{align}
A_{1}& = e^{-nT}\sum_{\bxt_{m}}P_{X}(\bX_{m}=\bx_{m})\sum_{\byt}f(\bx_{m},\by)\cdot\Pr\left\{\left.\by\in\hat{{\cal R}}_{m}^{c}\right|\bX_{m}=\bx_{m} \mbox{ transmitted}\right\}\label{lastTerm_training0}\\
& = e^{-nT}\sum_{\bxt_{m}}P_{X}(\bX_{m}=\bx_{m})\sum_{\bar Q}\sum_{\byt:\; \hat P_{\bar \bxt\bar \byt}=\bar Q}f(\bx_{m},\by)\cdot\Pr\left\{\left.\by\in\hat{{\cal R}}_{m}^{c}\right|\bX_{m}=\bx_{m} \mbox{ transmitted}\right\} \\
& \doteq e^{-nT}\max_{\bar Q}\sum_{\bxt_{m}}P_{X}(\bX_{m}=\bx_{m})\sum_{\byt:\; \hat P_{\bar \bxt\bar \byt}=\bar Q}f(\bx_{m},\by)\cdot\Pr\left\{\left.\by\in\hat{{\cal R}}_{m}^{c}\right|\bX_{m}=\bx_{m} \mbox{ transmitted}\right\} \\
& \doteq e^{-nT}\max_{\bar Q}\max_{\tilde Q}\sum_{\bxt_{m}}P_{X}(\bX_{m}=\bx_{m})\sum_{\byt:\; \hat P_{\bar \bxt\bar \byt}=\bar Q,\; \hat P_{\bxt_m \byt}=\tilde Q}f(\bx_{m},\by)\cdot\Pr\left\{\left.\by\in\hat{{\cal R}}_{m}^{c}\right|\bX_{m}=\bx_{m} \mbox{ transmitted}\right\} \\
& \doteq e^{-nT}\max_{\bar Q}\max_{\tilde Q}\sum_{\bxt_{m}}P_{X}(\bX_{m}=\bx_{m})\sum_{\byt:\; \hat P_{\bar \bxt\bar \byt}=\bar Q,\; \hat P_{\bxt_m \byt}=\tilde Q}f(\bx_{m},\by)\cdot\Pr\left\{\left.\by\in\hat{{\cal R}}_{m}^{c}\right|\bX_{m}=\bx_{m} \mbox{ transmitted}\right\}. \label{lastTerm_training}
\end{align}
\hrulefill
\vspace*{4pt}
\end{figure*}
Now, if the joint type of the training sequence $\bar{\bx}$ and the training part of $\by$ is $\bar Q$, and the type of the entire codeword $\bx_m$ and $\by$ is $Q$, then the type of the last $(1-\alpha)n$ symbols of $\bx_m$ and $\by$ is $\frac {Q -\alpha \bar Q}{1-\alpha}$. So, the probability in Eq. \eqref{probability of joint type} should now be replaced by
\begin{align}
p\doteq \exp\ppp{-n\cdot J(Q,\bar Q)}.
\end{align}
Consequently, for $Q\in\calS(\hat{P}_{\byt})$, we have \eqref{largeDe_training}, and we define \eqref{Udef_training}, both shown at the top of the next page.
\begin{figure*}[!t]
\normalsize
\setcounter{MYtempeqncnt}{\value{equation}}
\setcounter{equation}{13} 
\begin{align}
&\Pr\left\{ N_{\by}(Q)\geq e^{n\Omega(Q,\tilde{Q})}\right\} \doteq 
\begin{cases}
\exp\left\{ -n\left|J(Q,\bar Q)-R\right|^{+}\right\} \ & \Omega(Q,\tilde{Q})\leq0\\
1 \ &0<\Omega(Q,\tilde{Q})\leq R-J(Q,\bar Q)\\
0 \ & \Omega(Q,\tilde{Q})>R-J(Q,\bar Q)
\end{cases}\label{largeDe_training}\\
&U(\tilde{Q},\bar{Q})\triangleq \max_{Q\in\calS(\tilde{Q}_Y)}\begin{cases}
\exp\left[ -n(J(Q,\bar Q)-R)\right] \ & \Omega(Q,\tilde{Q})\leq0,\;J(Q,\bar Q)>R\\
1 \ &J(Q,\bar Q)\leq R,\;\Omega(Q,\tilde{Q})\leq R-J(Q,\bar Q)\\
0 \ &\text{otherwise}
\end{cases}\label{Udef_training}
\end{align}
\hrulefill
\vspace*{4pt}
\end{figure*}
Thus,  using the same derivation as in \eqref{maxMes}, but with \eqref{largeDe_training} replacing \eqref{largeDe}, we may continue \eqref{lastTerm_training} as follows:
\begin{align}
&A_{1} \doteq  e^{-nT}\max_{\bar Q}\max_{\tilde{Q}}\exp\pp{\alpha nH_{Y\vert X}(\bar{Q})}\nonumber\\
&\cdot\exp\pp{(1-\alpha)n H_{Y\vert X}\p{\frac {Q -\alpha \bar Q}{1-\alpha}}}
\exp\left[n G(\tilde{Q})\right] U(\tilde{Q},\bar{Q}).
\end{align}

Thus, we obtain that the exponent of $A_{1}$ is given by
\[
\lim_{n\to\infty}\frac{1}{n}\log A_{1}= -T-\min\ppp{\tilde{E}_{a}(R,T,\xi),\tilde{E}_{b}(R,T,\xi)},
\]
in which
\begin{align}
&\tilde{E}_{a}(R,T,\xi)\triangleq\min_{\bar Q} \min_{(Q,\tilde{Q})\in\calQ(\bar Q)}\left[\vphantom{H_{Y\vert X}\p{\frac {Q -\alpha \bar Q}{1-\alpha}}} -\alpha H_{Y\vert X}(\bar{Q}) \right.\nonumber\\
&\left.- (1-\alpha) H_{Y\vert X}\p{\frac {Q -\alpha \bar Q}{1-\alpha}}-G(\tilde{Q})+J(Q,\bar Q)-R\right]\label{Ea0}
\end{align}
where $\calQ(\bar Q)$ is defined in \eqref{Qdef_training}, and 
\begin{align}
&\tilde{E}_{b}(R,T,\xi)\triangleq\min_{\bar Q}\min_{\tilde{Q}\in{\cal L(\bar Q)}}\left[-\alpha H_{Y\vert X}(\bar{Q})\right.\nonumber\\
&\left.\ \ \ \ \ \ \ \ \ \ \ \ \ - (1-\alpha) H_{Y\vert X}\p{\frac {Q -\alpha \bar Q}{1-\alpha}}-G(\tilde{Q})\right]
\end{align}
where $\calL(\bar Q)$ is defined in \eqref{Ldef_training}. Now, we want to find the maximal $\xi$ for which 
\begin{align}
-T-\tilde{E}_{a}(R,T,\xi)\leq0,
\end{align}
\begin{align}
-T-\tilde{E}_{b}(R,T,\xi)\leq0.
\end{align}
The expressions for $\tilde{E}_{a}(R,T,\xi)$ and $\tilde{E}_{b}(R,T,\xi)$ can be simplified just as in the proof of Theorem \ref{Th:1} (see Eqs. \eqref{maxtht} and \eqref{maxtht2}). This results the conditions appearing in the theorem. 
\end{proof}

\section*{Acknowledgment}

The authors would like to thank the associate editor, Jun Chen, and the anonymous referees for their suggestions which helped improving the content of this paper.

\ifCLASSOPTIONcaptionsoff
  \newpage
\fi
\bibliographystyle{IEEEtran}
\bibliography{strings}

\begin{thebibliography}{10}
\providecommand{\url}[1]{#1}
\csname url@samestyle\endcsname
\providecommand{\newblock}{\relax}
\providecommand{\bibinfo}[2]{#2}
\providecommand{\BIBentrySTDinterwordspacing}{\spaceskip=0pt\relax}
\providecommand{\BIBentryALTinterwordstretchfactor}{4}
\providecommand{\BIBentryALTinterwordspacing}{\spaceskip=\fontdimen2\font plus
\BIBentryALTinterwordstretchfactor\fontdimen3\font minus
  \fontdimen4\font\relax}
\providecommand{\BIBforeignlanguage}[2]{{%
\expandafter\ifx\csname l@#1\endcsname\relax
\typeout{** WARNING: IEEEtran.bst: No hyphenation pattern has been}%
\typeout{** loaded for the language `#1'. Using the pattern for}%
\typeout{** the default language instead.}%
\else
\language=\csname l@#1\endcsname
\fi
#2}}
\providecommand{\BIBdecl}{\relax}
\BIBdecl

\bibitem{Goppa}
V.~D. Goppa, ``Nonprobabilistic mutual information without memory,''
  \emph{Probl. Cont. Information Theory}, vol.~4, pp. 97--102, 1975.

\bibitem{ZivUni}
J.~Ziv, ``Universal decoding for finite-state channels,'' \emph{IEEE Trans.
  Inf. Theory}, vol. IT-31, no.~4, pp. 453--460, July 1985.

\bibitem{Csis2}
I.~Csisz\'ar, ``Linear codes for sources and source networks: error exponents,
  universal coding,'' \emph{IEEE Trans. Inf. Theory}, vol. IT-28, no.~4, pp.
  585--592, July 1982.

\bibitem{csiszar2011information}
I.~Csisz{\'a}r and J.~K{\"o}rner, \emph{Information Theory: Coding Theorems for
  Discrete Memoryless Systems}.\hskip 1em plus 0.5em minus 0.4em\relax
  Cambridge University Press, 2011.

\bibitem{NeriUni}
N.~Merhav, ``Universal decoding for memoryless {G}aussian channels with a
  deterministic interference,'' \emph{IEEE Trans. Inf. Theory}, vol.~39, no.~4,
  pp. 1261--1269, July 1993.

\bibitem{UniNeri2}
------, ``Universal decoding for arbitrary channels relative to a given class
  of decoding metrics,'' \emph{IEEE Trans. Inf. Theory}, vol.~59, no.~9, pp.
  5566--576, Sep. 2013.

\bibitem{FerderLapidoth}
M.~Feder and A.~Lapidoth, ``Universal decoding for channels with memory,''
  \emph{IEEE Trans. Inf. Theory}, vol.~44, no.~5, pp. 1726--1745, Sep. 1998.

\bibitem{Forney68}
G.~D. Forney, Jr., ``Exponential error bounds for erasure, list, and decision
  feedback schemes,'' \emph{IEEE Trans. Inf. Theory}, vol.~14, no.~2, pp.
  206--220, 1968.

\bibitem{Burnashev}
M.~V. Burnashev, ``Data transmission over a discrete channel with feedback,''
  \emph{Problems of Information Transmission}, pp. 250--265, 1976.

\bibitem{ShluRate}
N.~Shulman, ``\emph{Communication over an unknown channel via common
  broadcasting},'' Ph.D. dissertation, Tel-Aviv University, 2003,
  \url{http://www.eng.tau.ac.il/~shulman/papers/Nadav\_PhD.pdf}.

\bibitem{rate1}
S.~Draper, B.~J. Frey, and F.~R. Kschischang, ``Rateless coding for non-ergodic
  channels with decoder channel state information,'' \emph{IEEE Trans. Inf.
  Theory}, vol.~55, no.~9, pp. 4119--4133, 2009.

\bibitem{rate2}
U.~Erez, G.~W. Wornell, and M.~D. Trott, ``Rateless space-time coding,'' in
  \emph{Proc. ISIT 2005}, Sep. 2005, pp. 1937--1941.

\bibitem{rate3}
J.~Jiang and K.~R. Narayanan, ``Multilevel coding for channels with non-uniform
  inputs and rateless transmission over the bsc,'' in \emph{Proc. ISIT 2006},
  2006, pp. 518--522.

\bibitem{rate4}
A.~Tchamkerten and E.~I. Telatar, ``Variable length codes over unknown
  channels,'' \emph{IEEE Trans. Inf. Theory}, vol.~52, no.~5, pp. 2126--2145,
  2006.

\bibitem{universal_minimax_erasure}
N.~Merhav and M.~Feder, ``Minimax universal decoding with an erasure option,''
  \emph{IEEE Trans. Inf. Theory}, vol.~53, no.~5, pp. 1664--1675, May. 2007.

\bibitem{Moulin_universal_erasure}
P.~Moulin, ``A {N}eyman-{P}earson approach to universal erasure and list
  decoding,'' \emph{IEEE Trans. Inf. Theory}, vol.~55, no.~10, pp. 4462--4478,
  2009.

\bibitem{merFeder}
M.~Feder and N.~Merhav, ``Universal composite hypothesis testing: a competitive
  minimax approach,'' \emph{IEEE Trans. Inf. Theory special issue in memory of
  Aaron D. Wyner}, vol.~48, no.~6, pp. 1504--1517, June 2002.

\bibitem{exact_erasure}
A.~Somekh-Baruch and N.~Merhav, ``Exact random coding exponents for erasure
  decoding,'' \emph{IEEE Trans. Inf. Theory}, vol.~57, no.~10, pp. 6444--6454,
  2011.

\bibitem{NeriSW}
N.~Merhav, ``Erasure/list exponents for {S}lepian-{W}olf decoding,'' \emph{IEEE
  Trans. Inf. Theory}, vol.~60, no.~8, pp. 4463--4471, Aug. 2014.

\end{thebibliography}
\vspace{-0.8cm}
\begin{IEEEbiography}
{Wasim Huleihel} (S'14) received the B.Sc. and the M.Sc. degrees in electrical
engineering from the Ben-Gurion University of the Negev, Beer-Sheva,
Israel, in 2012 and 2013, respectively. Currently, he is working toward the
Ph.D. degree in electrical engineering at the Technion Institute of Technology,
Haifa, Israel. His research interests are in the areas of information theory, and
relationships between information theory, statistics, detection, and estimation.
\end{IEEEbiography}
\vspace{-1cm}
\begin{IEEEbiography}
{Nir Weinberger} (S'14) received the B.Sc. and M.Sc. degrees (both summa cum laude) from Tel-Aviv University, Tel-Aviv, Israel, in 2006 and 2009, respectively. From 2006 to 2013 he served as an algorithms Engineer in the Israeli Defense Forces. Currently, he is pursuing his Ph.D. degree at the Technion - Israel Institute of Technology, Haifa, Israel. His research interest is information theory, with emphasis on large deviations aspects in coding problems.
\end{IEEEbiography}
\vspace{-1cm}
\begin{IEEEbiography}
{Neri Merhav} (S'86--M'87--SM'93--F'99) was born 
in Haifa, Israel, on March 16, 1957. He received the
B.Sc., M.Sc., and D.Sc.\ degrees from the Technion, 
Israel Institute of Technology,
in 1982, 1985, and 1988, 
respectively, all in electrical engineering.

From 1988 to 1990 he was with AT\&T Bell Laboratories,
Murray Hill, NJ, USA. 
Since 1990 he has been with the 
Electrical Engineering Department
of the Technion, where
he is now the Irving Shepard Professor.
During 1994--2000 he was also serving 
as a consultant to the Hewlett--Packard 
Laboratories -- Israel (HPL-I). 
His research interests include
information theory, statistical communications,
and statistical signal processing. He is especially
interested in the areas of lossless/lossy source coding
and prediction/filtering, relationships between information
theory and statistics, detection,
estimation, as well as 
in the area of Shannon Theory, including topics in
joint source--channel coding, source/channel simulation,
and coding with side information with applications to
information hiding and watermarking systems.
Another recent research interest concerns the relationships between
Information Theory and statistical physics.

Dr.\ Merhav was a co-recipient of the 1993 Paper Award 
of the IEEE Information Theory Society and he is a
Fellow of the IEEE since 1999. He also received the 
1994 American Technion Society 
Award for Academic Excellence and the 2002
Technion Henry Taub Prize for Excellence in Research.
From 1996 until 1999 he served as an Associate Editor
for Source Coding to the 
{\sc IEEE Transactions on Information Theory}.
He also served as a co--chairman of the Program Committee
of the 2001 IEEE International Symposium on Information Theory.
He is currently on the Editorial Board of {\sc Foundations
and Trends in Communications and Information Theory}.
\end{IEEEbiography}

\end{document}